\documentclass[sort, compress, 3p, times]{elsarticle}
\usepackage{xcolor}

\usepackage{amssymb}
\usepackage{amsmath}
\usepackage{amsthm}
\usepackage{booktabs}
\usepackage{subcaption}
\usepackage{algorithm}
\usepackage{algpseudocode}
\usepackage{soul}
\usepackage{orcidlink}
\usepackage{natbib}
\makeatletter
\def\ps@pprintTitle{%
 \let\@oddhead\@empty
 \let\@evenhead\@empty
 \def\@oddfoot{\footnotesize\itshape
       Accepted for publication in \ifx\@journal\@empty Elsevier
       \else\@journal\fi\hfill\today}%
 \let\@evenfoot\@oddfoot}
\makeatother

\journal{Information Sciences}

\begin{document}

\begin{frontmatter}



\title{The Theory and Practice of Computing the Bus-Factor}

\author[demacs]{Sebastiano A. Piccolo\corref{cor1}}
\cortext[cor1]{Corresponding author}
\ead{sebastiano.piccolo@unical.it}
\affiliation[demacs]{%
  organization = {DeMaCS, University of Calabria},
  addressline = {Viale P. Bucci},
  city = {Rende, Cosenza},
  postcode = {87036},
  state = {Italy}
}

\author[unime]{Pasquale De Meo}
\ead{pdemeo@unime.it}
\affiliation[unime]{%
  organization = {DICAM, University of Messina},
  postcode = {98122},
  city = {Messina},
  state = {Italy}
}

\author[demacs]{Giorgio Terracina}
\ead{giorgio.terracina@unical.it}

\author[demacs]{Gianluigi Greco}
\ead{gianluigi.greco@unical.it}

\newtheorem{theorem}{Theorem}
\newtheorem{proposition}{Proposition}
\newtheorem{corollary}{Corollary}
\newtheorem{lemma}{Lemma}[theorem]
\newtheorem{definition}{Definition}

\begin{abstract}
The \textit{bus-factor} is a measure of project risk with respect to personnel availability, informally defined as the number of people whose sudden unavailability would cause a project to stall or experience severe delays.
Despite its intuitive appeal, existing bus-factor measures rely on heterogeneous modeling assumptions, ambiguous definitions of failure, and domain-specific artifacts, limiting their generality, comparability, and ability to capture project fragmentation.
In this paper, we develop a unified, domain-agnostic framework for bus-factor estimation by modeling projects as bipartite graphs of people and tasks and casting the computation of the bus-factor as a family of combinatorial optimization problems. 
Within this framework, we formalize and reconcile two complementary interpretations of the bus-factor, redundancy and criticality, corresponding to the Maximum Redundant Set and the Minimum Critical Set, respectively, and prove that both formulations are NP-hard.
Building on this theoretical foundation, we introduce a novel bus-factor measure inspired by network robustness. 
Unlike prior approaches based solely on task coverage, the proposed measure captures both loss of coverage and increasing project fragmentation by tracking the largest connected set of tasks under progressive contributor removal. 
The resulting measure is normalized, threshold-free, and applicable across domains; we show that its exact computation is NP-hard as well.
We further propose efficient linear-time approximation algorithms for all considered measures. 
Finally, we evaluate their behavior through a sensitivity analysis based on controlled perturbations of project structures, guided by expectations derived from project management theory. 
Our results show that the robustness-based measure behaves consistently with these expectations and provides a more informative and stable assessment of project risk than existing alternatives.
\end{abstract}




\end{frontmatter}



\section{Introduction}
\label{sec:intro}

The \emph{bus-factor} (also known as \textit{truck-factor}) of a project is informally defined as the number of people who have to disappear -- as if they were hit by a bus -- from a project until the project stalls or experiences a significant delay, either because of a lack of knowledgeable personnel~\cite{zazworka2010developers} or because of high project fragmentation~\cite{piccolo2018design}.

Many projects have been found to exhibit a highly skewed labor division, with a small fraction of people being responsible for most of the total workload~\cite{yamashita2015pareto}.
Therefore, even the unavailability of a small number of professionals can have serious consequences for project completion.
As such, the bus-factor captures a real risk for organizations and projects, and its practical importance has been widely recognized~\cite{jabrayildaze2022bus}. 

Despite its intuitive appeal, computing the bus-factor involves several nontrivial design choices, including how to model projects, how to represent contributor knowledge, and how to define project stalling~\cite{ricca2011difficulty}.
Existing approaches differ in the underlying conceptualization of the bus-factor.
In particular, the bus-factor is understood either in terms of redundancy, quantifying how many people can safely leave a project without harming it~\cite{zazworka2010developers}, or in terms of criticality, quantifying the smallest set whose departure would compromise the project~\cite{avelino2016novel}.

Such a conceptual difference can lead to markedly different bus-factor estimates for the same project~\cite{ferreira2017comparison}.
In addition, existing approaches depend on the use they make of domain-specific artifacts, notably GitHub metadata: some approaches consider only commits~\cite{avelino2016novel}, some consider file ownership~\cite{zazworka2010developers}, some use the whole history~\cite{cosentino2015assessing} and others give higher weight to more recent commits or contributions~\cite{jabrayildaze2022bus}. 
This heterogeneity makes it difficult to evaluate and interpret the different approaches to bus-factor estimations.

As such, in this paper we develop a unified theoretical framework that allows us to reason about the cascading effects that arise when contributors leave a project, decoupling the computation of the bus-factor from domain-specific artifacts.
We adopt a minimal yet expressive representation in which a project is modeled as a bipartite graph connecting people to tasks, and condense prior bus-factor definitions as two combinatorial problems on bipartite graphs: Maximum Redundant Set (MRS) and Minimum Critical Set (MCS) that capture respectively the redundancy and the criticality interpretations of the bus-factor.

This framework reveals that prior approaches, which rely on task coverage and fixed thresholds to define project failure, are unable to capture project fragmentation and the role of integrators~\cite{piccolo2018design} -- contributors whose absence disconnects otherwise independent components of a project. 
Fig.~\ref{fig:summary-of-the-paper} illustrates this weakness.

Consider the toy project represented as a bipartite graph of people and tasks (Fig.~\ref{fig:summary-of-the-paper}\textbf{A}).
The project is structured to depend on a single integrator, person $p_1$, who keeps together the four project modules defined by each task $t_1,\dots,t_4$.
Intuitively, once $p_1$ leaves, the project fractures into four isolated modules of one task and two people each.
In this fragmented state, the unavailability of even one more person puts the project at risk, as a task might then rely on only a single contributor.
Therefore, the bus-factor of this project should be between two and three people.

However, if we apply any of the aforementioned definitions, we find that MRS estimates the bus-factor to be 8 people (Fig.~\ref{fig:summary-of-the-paper}\textbf{B}), while MCS finds a bus-factor of 7 people (Fig.~\ref{fig:summary-of-the-paper}\textbf{C}).
This overestimation is a direct consequence of both the stalling condition defined via a threshold and the fact that coverage-based measures fail to account for integrators. 
This highlights the limitations of coverage-based measures and thresholds, and motivates the need for a bus-factor formulation that can account for project fragmentation as contributors leave.

\begin{figure}[t]
    \centering
    \includegraphics[width=\textwidth]{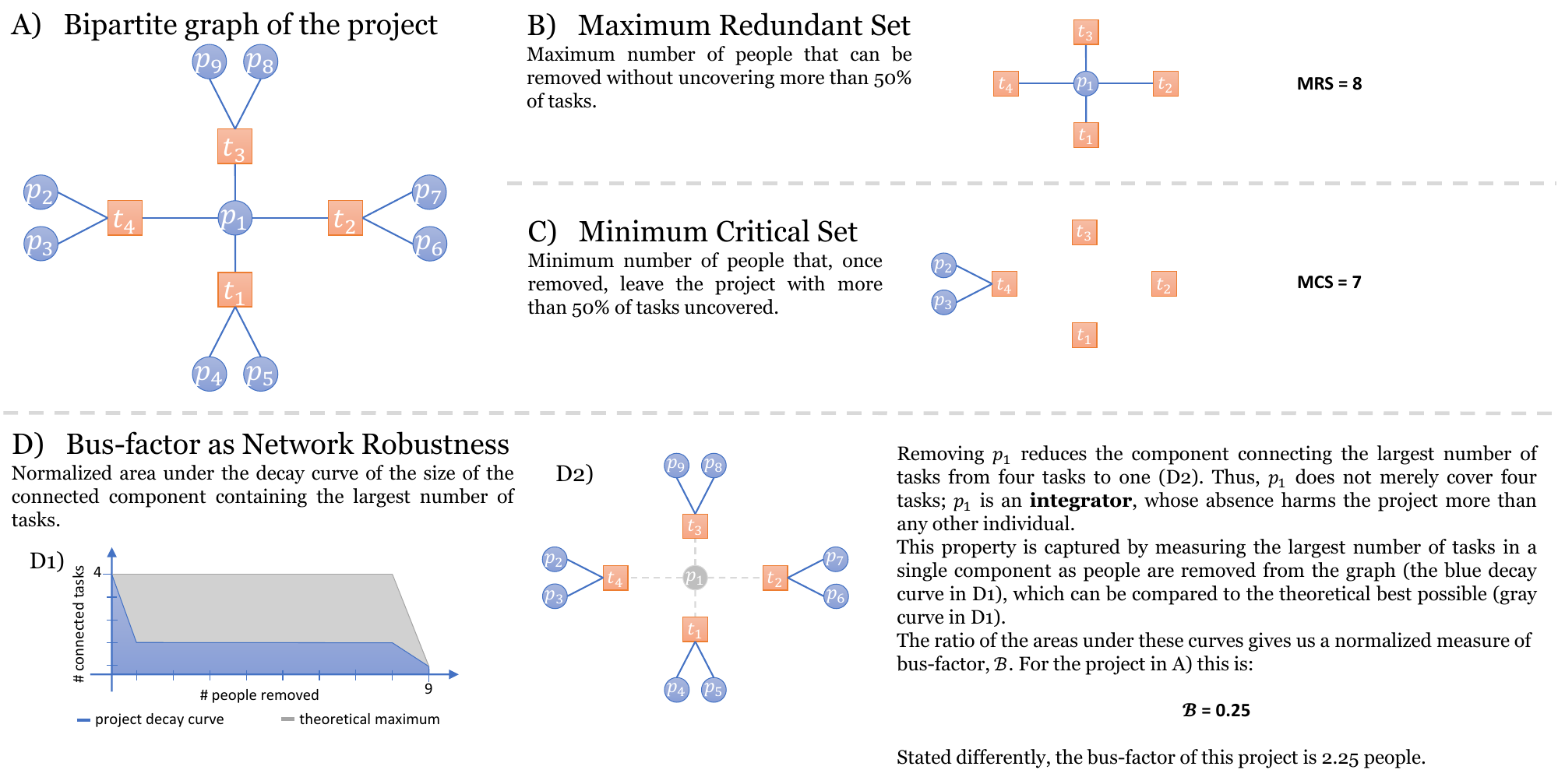}
    \caption{Visual comparison of the measures to estimate a project bus-factor. Consider the toy project in \textbf{A)}. When $p_1$ leaves, the project fragments into four disconnected components. Under further removals, tasks quickly become critically dependent on single contributors. \textbf{B)} Maximum Redundant Set (MRS) estimates a bus-factor of 8 by removing redundant contributors. \textbf{C)} Minimum Critical Set (MCS) estimates a bus-factor of 7 by applying a coverage threshold, ignoring project topology. \textbf{D)} Our robustness-based measure captures fragmentation by evaluating the largest number of tasks that remain connected in a single component as people leave; yielding a normalized, threshold-free bus-factor consistent with project intuition.}
    \label{fig:summary-of-the-paper}
\end{figure}

In this paper, we develop such a measure (hereafter referred to as \emph{Robustness}) by taking inspiration from network robustness and replacing task coverage with the largest number of tasks connected to a single component (Fig.~\ref{fig:summary-of-the-paper}\textbf{D}).
Furthermore, we avoid the need for fixed thresholds by computing the area under the decay curve of the maximum number of connected tasks as people are progressively removed from the project graph (Fig.~\ref{fig:summary-of-the-paper}\textbf{D1}).
We normalize this area by dividing it by the theoretical maximum, obtained on a fully connected bipartite graph with the same number of people and tasks.
Our measure gives us a normalized score of $0.25$, or equivalently a bus-factor of $2.25$ people, congruent with what we expected from our intuition.

Beyond introducing a unified theoretical framework for the computation of the bus-factor and a new measure, this paper provides a theoretical treatment of bus-factor definitions, establishes their computational complexity, proposes scalable approximation algorithms, and evaluates their practical behavior through a systematic experimental and sensitivity analysis guided by project management theory.
We hope that this paper can serve as a starting point for researchers and practitioners interested in the study of bus-factor measures.
To facilitate this process and future research in the field, we make our code and generated data publicly available \footnote{\url{https://github.com/sapiccolo/BusFactorX}}.

\paragraph{Summary of Contributions}
\label{sub:contributions}

We model a project as a bipartite graph $G = (P, T, E)$, where $P$ is the set of $n$ people and $T$ is the set of $m$ tasks, and cast the computation of the bus-factor as a family of combinatorial problems on such a graph, decoupling its definition from domain-specific artifacts.

In Section~\ref{sub:framework_prior}, we develop a unified, domain-agnostic theoretical framework by formalizing prior approaches and establishing an equivalence between redundancy-based and criticality-based interpretations of the bus-factor (Proposition~\ref{prop:equivalence_za}). 
This insight allows us to recast existing measures into two complementary formulations: the Maximum Redundant Set (MRS), identifying the largest set of contributors that can be removed without compromising project progress, and the Minimum Critical Set (MCS), identifying the smallest set whose removal stalls the project. 
We prove that both formulations are NP-hard (Theorems~\ref{thm:critical_set} and~\ref{thm:partial_set_cover}).

Building on this framework, in Section~\ref{sub:bf_net_robustness}, we introduce a new bus-factor measure (Robustness) inspired by network robustness, which evaluates the size of the largest connected set of tasks as contributors are progressively removed from the project graph. 
By defining the bus-factor as the normalized area under the resulting decay curve, the measure captures both loss of task coverage and increasing project fragmentation. 
We prove that this formulation is also NP-hard (Theorem~\ref{thm:network_robustness}), thereby closing an open problem posed by Zhao \textit{et al.}~\cite{zhao2021towards}.

We then develop efficient linear-time approximation algorithms for all three measures and analyze their worst-case approximation guarantees under degree-based removal strategies (Section~\ref{sec:practice}). 
In addition, we empirically evaluate their performance by comparing degree-based removal to standard centrality-based strategies commonly used in robustness analysis (Section~\ref{sec:performance}).

Finally, in Section~\ref{sec:experiments}, we perform a sensitivity analysis to assess the suitability of MRS, MCS, and Robustness as bus-factor indicators. 
We introduce a suite of tests that simulate intervention strategies on project structures, grounded in expectations derived from project management theories. 
Our results show that Robustness aligns consistently with these expectations and offers superior sensitivity and practical relevance, while MRS proves overly optimistic and MCS exhibits counterintuitive behavior due to its inability to capture project fragmentation.

\section{Related Work}
\label{sec:literature}

The bus-factor has been studied across several research communities, most prominently in software engineering, where it is used to assess the vulnerability of projects to the loss of key contributors. 
Despite its intuitive definition, prior work differs substantially in how projects are modeled, how contributor knowledge is inferred, and how stalling is defined.
In this section, we organize the literature around conceptual assumptions, with the goal of clarifying the strengths and limitations of existing approaches and motivating the need for a domain-agnostic formulation.

\paragraph{Artifact-based and GitHub-specific Bus-Factor Measures}

Most empirical studies of the bus-factor rely on data extracted from version control systems, particularly GitHub repositories.
In these approaches, tasks are typically equated with files, and contributor knowledge is inferred from historical activity such as commits, edits, or code ownership.

Zazworka \textit{et al.}~\cite{zazworka2010developers} introduced a family of coverage-based measures, defining a project as safe if a given percentage of files remains covered after contributor removal.
While conceptually appealing, this approach does not explicitly model developer knowledge~\cite{ricca2011difficulty} and becomes computationally infeasible for projects with more than a few dozen contributors~\cite{hannebauer2014algorithmic}.

Subsequent work refined how contributor knowledge is inferred.
Cosentino \textit{et al.}~\cite{cosentino2015assessing} introduced primary and secondary developers based on contribution thresholds, while Avelino \textit{et al.}~\cite{avelino2016novel} proposed a heuristic that iteratively removes the most knowledgeable developer, measured via Degree of Authorship (DoA), until more than 50\% of files become abandoned.
This heuristic has since been adopted in several follow-up studies~\cite{jabrayildaze2022bus}.

These approaches share common assumptions: tasks correspond to software artifacts, contributor knowledge is inferred from repository activity, and project stalling is determined via fixed thresholds.

\paragraph{Toward a Domain-Agnostic Bus-Factor Measure}

The limitations of artifact-based approaches have motivated more abstract formulations of the bus-factor that focus on the structure of the assignment of people to tasks.
Piccolo \textit{et al.}~\cite{piccolo2024evaluating} proposed modeling projects as bipartite graphs connecting contributors to tasks, effectively decoupling the estimation of contributions from the computation of the bus-factor.
This abstraction enables the application of bus-factor analysis across a wide range of collaborative settings beyond software engineering.
In the same work, they defined the bus-factor leveraging concepts from network robustness, and used their measure to increase the bus-factor of a real-world biomass power plant project.

However, their work does not provide a unified theoretical treatment of prior definitions, nor does it offer a systematic comparison with prior measures.
In addition, they do not analyze the limitations of task coverage-based formulations in capturing project fragmentation and the role of contributors who integrate multiple components of a project.

\paragraph{Network Robustness and Project Resilience}
The problem of assessing project vulnerability under contributor loss is closely related to the study of network robustness, which examines how the functionality of a network degrades as nodes or edges are removed. Seminal work by Albert \textit{et al.}~\cite{albert2000error} showed that networks with heterogeneous degree distributions are resilient to random failures but highly vulnerable to targeted attacks on nodes with high degree. 
This behavior has been observed across markedly different systems, including the Internet~\cite{albert2000error}, criminal networks~\cite{agreste2016networks}, and engineering systems~\cite{piccolo2018design}.
This generality makes network robustness the ideal starting point to build a theory of bus-factor estimation.

Theoretical efforts have explained the roots of this general behavior through percolation theory~\cite{callaway2000network} and subsequent studies introduced robustness measures based either on the size of the largest connected component under progressive node removal~\cite{schneider2011mitigation} or on the graph spectrum~\cite{shang2013measuring}.
More recent work has explored robustness in bipartite networks~\cite{shang2013measuring}, as well as percolation phenomena in hypergraphs, including the role of so-called anchor nodes in maintaining global connectivity~\cite{shang2025percolation}.

Despite these advances, classical robustness measures are not directly applicable to bus-factor estimation.
They typically treat all nodes as equivalent and focus on node reachability, whereas project networks exhibit an inherent asymmetry between contributors and tasks.
Moreover, robustness measures are rarely normalized in a way that enables meaningful comparisons across projects of different sizes.

Bridging the gap between network robustness and bus-factor estimation requires a formulation that respects the bipartite nature of projects, interprets robustness in terms of personnel loss, and captures both task coverage and project fragmentation.

\paragraph{Synthesis}
In summary, existing bus-factor measures either rely on domain-specific artifacts and threshold-based notions of project stalling, or adopt abstract formulations that lack a unified theoretical characterization and comparability across measures. 
Although network robustness provides a powerful conceptual framework for reasoning about the resilience of complex systems, classical robustness measures are not directly applicable to project settings.

Our work addresses these limitations by modeling projects as bipartite graphs of people and tasks and by interpreting the bus-factor through the lens of network robustness. 
Our framework unifies prior definitions based on redundancy or criticality, establishes their computational complexity, and introduces a normalized robustness-based measure that captures both task coverage and project fragmentation in a domain-agnostic manner.

\section{Preliminaries}
\label{sec:preliminaries}

Networks are a convenient mathematical formalism to represent and analyze complex systems, focusing on the way the components of such systems are connected. 
Components are represented as vertices, and their connections are represented by edges between vertices.
Graph methods are flexible enough to enable the representation and analysis of heterogeneous systems, with multiple layers of connections~\cite{kivela2014multilayer}, which evolve over time~\cite{holme2012temporal}. 

A graph is a pair $G = (V, E)$ where $V$ is the set of vertices and $E \subseteq V \times V$ is the set of edges. 
A graph is represented through its adjacency matrix $A$, with $A_{ij} = 1$ if vertices $i$ and $j$ are connected and $A_{ij} = 0$ otherwise. 
The degree $k_i$ of a vertex $i$ is the number of vertices connected to it; that is, $k_i = \sum_j A_{ij}$.

A \emph{subgraph} of a graph $G$ is the graph induced by a subset $S \subseteq V$ of vertices of $G$, and it is denoted by $G[S]$.
A \emph{connected component} of a graph is a maximal connected subgraph.
A fully connected subgraph is called a \emph{clique}. 
A set of nodes that includes at least one endpoint of each edge of $G$ is called a \emph{vertex cover}.

If the set of vertices $V$ can be divided into two sets $V_1$ and $V_2$ with $V_1 \cap V_2 = \emptyset$ and $V_1 \cup V_2 = V$, such that $\forall (i,j) \in E \; i \in V_1 \wedge j \in V_2$, then the graph is called \emph{bipartite}. 
In the remainder of this paper, we denote a bipartite graph as a tuple $\mathcal{G} = (V_1, V_2, E)$.

To establish the NP-hardness of computing the bus-factor under various formulations, we use reductions from known NP-hard problems. 
We utilize the following problems, formally defined below. Proofs of their NP-hardness can be found in Garey and Johnson~\cite{garey1979computers}.

\begin{description}
    \item[Set Cover] Given a universe $U$ of $n$ elements, a family $\mathcal{F} = \{ X_1, X_2, \dots, X_m \}$ of $m$ subsets of $U$ such that $\bigcup_i X_i = U$, and a positive integer $k \le m$, is there a subset $S \subseteq \mathcal{F}$ with $|S| \le k$ such that $\bigcup_{s \in S} s = U$?
    \item[Clique] Given an undirected graph $G = (V, E)$ and a positive integer $k \le |V|$, does $G$ contain a clique of size at least $k$?
    \item[Vertex Cover] Given an undirected graph $G = (V, E)$ and a positive integer $k \le |V|$, does $G$ contain a vertex cover of size at most $k$?
\end{description}

\section{The Theory of Computing the Bus-Factor}
\label{sec:theory}
In this section, we introduce a theoretical framework to understand and model the computation of the bus-factor of a project. 
Our theoretical contribution is twofold.
First, we systematize existing approaches as combinatorial problems on bipartite graphs under two complementary formulations, namely the Maximum Redundant Set and the Minimum Critical Set.
Second, with the flaws of existing approaches in mind, we formulate a novel bus-factor measure inspired by network robustness, which addresses these shortcomings.
For all proposed bus-factor formulations, we prove their NP-hardness. 

Our starting point is the abstraction of a project as a bipartite graph $G = (P, T, E)$, where $P$ is the set of $n$ people, $T$ is the set of $m$ tasks, and $E$ is the set of edges that represent the assignment of people to tasks.
In certain projects, such as open-source projects with many voluntary contributors, not every person contributes meaningfully to every file they edit.
For instance, a person might only update the license details of a file.
To address this issue, domain-informed or domain-specific graph filtering strategies can be applied, in which an edge between a person and a task is removed if that person does not possess meaningful knowledge of the task.
For example, GitHub-specific heuristics known as Degree of Authorship (DoA)~\cite{jabrayildaze2022bus}, which are designed to identify meaningful contributors for each source code file, naturally become graph filtering strategies in our framework.

We now consider existing approaches to computing the bus-factor and formalize them as two combinatorial problems on bipartite graphs.

\subsection{Systematizing Existing Approaches as Combinatorial Problems on Bipartite Graphs}
\label{sub:framework_prior}

Zazworka \textit{et al.}~\cite{zazworka2010developers} were the first to consider the problem of computing the bus-factor and defined the following family of measures:
\begin{equation}
    \label{eq:zazworka}
    Z_{\alpha,t}(G) = \max \{k \in \mathbb{N} : \mathrm{cov}_\alpha(G, k) \ge t m\}
\end{equation}

with the aggregation function $\alpha \in \{ \mathrm{min, max, avg} \}$ and $0 < t \le 1$. 
Consequently, $\mathrm{cov}_\alpha(G, k)$ indicates the minimum, maximum, or average coverage of $G$ -- that is, the min/max/avg number of tasks that are still covered -- after the removal of any subset of $k$ people.
Since $Z_{\mathrm{avg},t}(G)$ was found to be inadequate~\cite{hannebauer2014algorithmic}, we do not consider it here and limit our analysis to $Z_{\mathrm{min},t}(G)$ and $Z_{\mathrm{max},t}(G)$.
Intuitively, Eq.~\ref{eq:zazworka} defines the bus-factor in terms of \emph{structurally redundant} people who can leave the project without harming it.
The higher the number of structurally redundant people, the higher the bus-factor.

Avelino \textit{et al.}~\cite{avelino2016novel}, by contrast, informally defined the bus-factor as the smallest number of people whose departure isolates more than a fraction $t$ of the tasks, where $t$ is a threshold defining when the project is considered stalled.
We formalize the definition from Avelino \textit{et al.} as the problem of finding a \textit{Minimum Critical Set}, defined as follows:
\begin{equation}
    \label{eq:avelino_bf}
    MCS_{t}(G) = \min_{X \subseteq P } |X|,\quad \phi(G[P \setminus X]) > t m
\end{equation}

where $0 < t < 1$ and $\phi(G)$ returns the number of isolated tasks in $G$.
Note that Eq.~\ref{eq:avelino_bf} captures the informal notion that the bus-factor is the minimum number of people that have to become unavailable in order for a project to stall.
The larger the minimum critical set, the higher the bus-factor.

We now prove the following relation between $Z_{\min,t}(G)$ and $MCS_{t}(G)$.

\begin{proposition}
    \label{prop:equivalence_za}
    $Z_{\min,t}(G) = MCS_{1-t}(G) - 1$, for all $0 < t < 1$.
\end{proposition}

\begin{proof}
    Let us recall that 
    \[
    Z_{\min,t}(G) = \max \{k \in \mathbb{N} : \mathrm{cov}_{\min}(G, k) \ge t m\}.
    \]
    The function $\mathrm{cov}_{\min}(G, k)$ returns the minimum number of tasks still covered by at least one person after removing $k$ people. That is,
    \[
    \mathrm{cov}_{\min}(G, k) = \min_{C \subseteq P,\; |C| = k} m - \phi(G[P \setminus C]).
    \]
   
    We can therefore rewrite $Z_{\min,t}(G)$ as follows:
    \begin{align*}
       Z_{\min,t}(G) 
       &= \max_{k \in \mathbb{N}} \min_{C \subseteq P,\; |C| = k} m - \phi(G[P \setminus C]) \ge t m \\
       &= \max_{k \in \mathbb{N}} \min_{C \subseteq P,\; |C| = k} \phi(G[P \setminus C]) \le m - t m \\
       &= \max_{k \in \mathbb{N}} \min_{C \subseteq P,\; |C| = k} \phi(G[P \setminus C]) \le m(1-t).
    \end{align*}
    Recall that $MCS_{1-t}(G)$ is the minimum number of people that isolates more than $m(1-t)$ tasks.
    That is, $MCS_{1-t}(G)$ is the smallest $k \in \mathbb{N}$ for which $\min_{C \subseteq P,\; |C| = k} \phi(G[P \setminus C]) > m(1-t)$.
    As such, $Z_{\min,t}(G) = MCS_{1-t}(G) - 1$.
\end{proof}

Note that $Z_{\min,t}(G)$ is also defined for $t = 1$, that is, the minimum number of people to remove such that the remaining coverage is less than $m$.
This case is not particularly interesting and is easily solvable, as it asks to isolate a single task with the minimum number of people, namely the task with the lowest degree.

We are now ready to prove that computing the Minimum Critical Set (MCS) is NP-hard.

\begin{theorem}
    \label{thm:critical_set}
    \textsc{Minimum Critical Set} is NP-hard.
\end{theorem}

\begin{proof}
    The \textsc{Minimum Critical Set} is an optimization problem. To show that it is NP-hard, we show that its decision version (i.e., a problem that admits a yes/no answer), which we call \textsc{Critical Set}, is NP-complete.
    In \textsc{Critical Set}, we are given a bipartite graph $G = (P, T, E)$, a positive integer $k$, and a real number $0 < t < 1$, and we ask whether there exists a set $X \subseteq P$ with $|X| \le k$ such that $\phi(G[P \setminus X]) > t m$.
    The problem \textsc{Critical Set} is in NP, since a candidate solution $X$ can be verified in polynomial time by computing the induced subgraph $G[P \setminus X]$, counting the number of isolated tasks, and checking whether this number exceeds $t m$.
    
    We show that \textsc{Critical Set} is NP-hard by reducing \textsc{Clique} to it. 
    In \textsc{Clique}, given a simple undirected graph $G' = (V', E')$ and a positive integer $k' < |V'|$, we ask whether there exists a fully connected subgraph of size at least $k'$.
    We reduce \textsc{Clique} to \textsc{Critical Set} as follows: we create a bipartite graph $G$ where $P = V', T = E'$, $E = \left\{ (a,b), a \in P, b \in T \quad \text{iff}\; a\; \text{is an endpoint of}\; b \right\}$, $k = k'$, and $t = \frac{k(k-1)-2}{2|T|}$ (Fig.~\ref{fig:Hardness_MCS_MRS}\textbf{A}).

    In this instance of \textsc{Critical Set}, we ask whether there exists a subset $C \subseteq P$ of size at most $k$ such that $\phi(G[P \setminus C]) > t|T| = \frac{k(k-1)-2}{2|T|}|T| = \frac{k(k-1)-2}{2}$.
    This instance of \textsc{Critical Set} has a positive answer if and only if the original \textsc{Clique} instance has a clique of size at least $k$.
    If $G'$ has a clique of size at least $k$, then the same set of nodes forms a solution to the corresponding instance of \textsc{Critical Set}.
    Indeed, a clique of $k$ nodes is the only subgraph on $k$ nodes that contains exactly $k(k-1)/2$ edges.
    Since in our bipartite construction we have $P = V'$ and $T = E'$, the same set of nodes in the clique of $G'$ corresponds to a subset $C \subseteq P$ such that $\phi(G[P \setminus C]) > \frac{k(k-1)-2}{2|T|}|T| > \frac{k(k-1)-2}{2}$.

    Conversely, if the bipartite graph $G$ admits a critical set $C$ of size at most $k$ such that $\phi(G[P \setminus C]) > \frac{k(k-1)-2}{2}$, then this set of $k$ nodes from $P$ isolates at least $k(k-1)/2$ nodes in $T$.
    However, by construction, and since $G'$ is simple and undirected, a subgraph of $k$ nodes in $G'$ can contain at most $k(k-1)/2$ edges.
    This implies that the set $C$ must consist of exactly $k$ nodes and isolate exactly $k(k-1)/2$ nodes in $T$.
    Therefore, $C$ corresponds to a clique of $k$ nodes in $G'$.
    
    As a result, \textsc{Clique} $\le_{\mathrm{p}}$ \textsc{Critical Set}, and hence \textsc{Critical Set} is NP-complete.
\end{proof}

\begin{figure}[t]
    \centering
    \includegraphics[width=\textwidth]{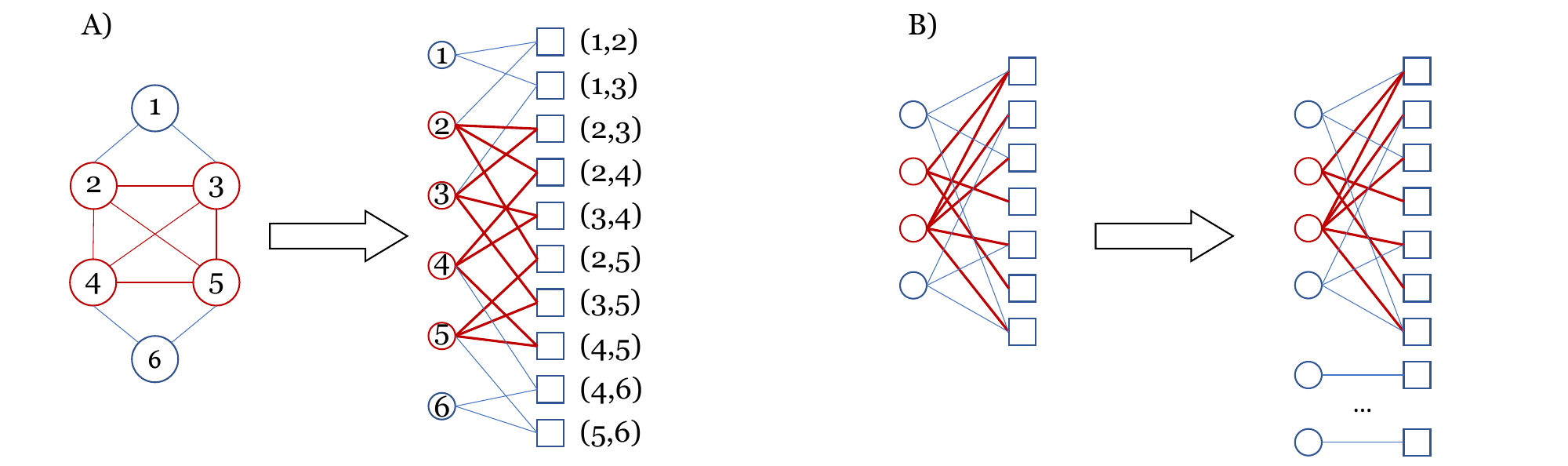}
    \caption{Visualizations of the NP-hardness reductions. \textbf{A}) Reduction from \textsc{Clique} to \textsc{Minimum Critical Set} (solutions highlighted): nodes in the original graph become the left vertex set ($P$) in the bipartite construction, and edges become the right vertex set ($T$). Edges in the bipartite graph represent the incidences of the original graph. A clique in the original graph corresponds to a minimum critical set in the bipartite graph. \textbf{B}) Reduction from \textsc{Set Cover} to \textsc{Partial Set Cover} for the \textsc{Maximum Redundant Set} proof: an instance of \textsc{Set Cover}, represented here as a bipartite graph with its solution highlighted, is transformed into an instance of \textsc{Partial Set Cover} by adding an appropriate number of bipartite vertex pairs.}
    \label{fig:Hardness_MCS_MRS}
\end{figure}

As a consequence of Proposition~\ref{prop:equivalence_za}, in addition to the \textsc{Minimum Critical Set} we only need to consider $Z_{\max,t}(G)$, which models a best-case scenario by identifying the maximum number of people who can safely become unavailable without uncovering (i.e., isolating) more than $t\%$ of tasks.
We formulate this quantity as the following \textsc{Maximum Redundant Set} problem on a bipartite graph:
\begin{equation}
    \label{eq:mrs}
    MRS_{t}(G) = \max_{X \subseteq P} |X| \quad \text{s.t.} \quad \phi(G[P \setminus X]) < t m,\quad 0 < t \le 1.
\end{equation}

Note that the MRS defined in Eq.~\ref{eq:mrs}, when $t = 1$, asks for a set which is the complement of the \textit{Minimum Set Cover}.
Set Cover admits a natural representation as a bipartite graph $G = (P, T, E)$, where the set $P$ corresponds to $\mathcal{F}$, the set $T$ corresponds to the universe $U$, and edges between $P$ and $T$ encode the coverage of each subset in $\mathcal{F}$: $E = \{ (i, j) \mid i \in P,\ j \in T,\ \text{and } j \in s_i \text{ for some } s_i \in \mathcal{F} \}$.

Similarly, for a threshold $0 < t < 1$, MRS asks for a set that is the complement of a \textit{Partial Set Cover}.
Given a universe $U = \{u_1, u_2, \dots, u_m\}$ of $m$ elements, a family $\mathcal{F} = \{S_1, S_2, \dots, S_n\}$ of $n$ subsets $S_i \subseteq U$ such that $\bigcup_i S_i = U$, a positive integer $k$, and a real number $0 < t < 1$, the Partial Set Cover problem asks whether there exists a subset $X \subseteq \mathcal{F}$ with $|X| \le k$ such that $\left| \bigcup_{X_i \in X} X_i \right| \ge t m$.
We now prove the following results.

\begin{theorem}
\label{thm:partial_set_cover}
\textsc{Partial Set Cover} is NP-complete.
\end{theorem}

\begin{proof}
    We show that \textsc{Partial Set Cover} is NP-complete by reducing \textsc{Set Cover} to it.
    The problem \textsc{Partial Set Cover} is in NP, since a candidate solution can be verified in polynomial time by counting the number of distinct elements it covers and checking whether this number is at least $t m$.

    Suppose we are given an instance of \textsc{Set Cover} with universe $U'$, family of subsets $\mathcal{F}'$, and a positive integer $k'$.
    Let $0 < t \le 1$. We transform this instance into an instance of \textsc{Partial Set Cover} as follows.
    We define the universe $U$ by adding $\lceil \frac{m}{t} \rceil - m$ new elements $e_i$ to $U'$, and we define the family $\mathcal{F}$ by adding $\lceil \frac{m}{t} \rceil - m$ subsets, each equal to $\{e_i\}$, to $\mathcal{F}'$.
    We also set $k = k'$ (Fig.~\ref{fig:Hardness_MCS_MRS}\textbf{B}).

    In this instance of \textsc{Partial Set Cover}, we ask whether there exists a subset $X \subseteq \mathcal{F}$ with $|X| \le k$ such that $\left| \bigcup_{X_i \in X} X_i \right| \ge tm$.
    This instance has a solution of size at most $k$ if and only if the original instance of \textsc{Set Cover} has a set cover of size at most $k'$.
    If the original instance of \textsc{Set Cover} admits a solution of size at most $k'$, then this solution is also a solution to the corresponding instance of \textsc{Partial Set Cover}.
    
    Conversely, suppose that the constructed instance of \textsc{Partial Set Cover} has a solution of size at most $k$ that covers at least $t m$ elements.
    Since the $\lceil \frac{m}{t} \rceil - m$ added subsets each cover exactly one element, they cannot contribute to a solution of size at most $k$ that covers at least $t m$ elements.
    Therefore, any such solution must consist solely of subsets from $\mathcal{F}'$ and hence corresponds to a set cover of size at most $k'$ in the original instance of \textsc{Set Cover}.
\end{proof}

\begin{corollary}
    \label{thm:redundant_set}
    \textsc{Maximum Redundant Set} is NP-hard.
\end{corollary}

\begin{proof}
    Since a Redundant Set is the complement of a Partial Set Cover, the proof follows from Theorem \ref{thm:partial_set_cover}.
\end{proof}

\subsection{Limitations of MCS and MRS}
\label{sub:coverage_limitations}

Both MCS and MRS ignore the underlying project structure by treating contributors solely in terms of task coverage.
This prevents them from capturing the role of integrators~\cite{piccolo2018design}, which, from a topological point of view, correspond to \textit{anchor nodes} in bipartite graphs (and hypergraphs), i.e., nodes whose removal fragments the network~\cite{shang2025percolation}.

Similarly, MCS and MRS assume that a task can be carried out as long as there is at least one person to cover it.
From a risk and project management perspective, this unrealistic simplification implies that it is possible to increase the bus-factor of a project simply by adding \emph{singletons} (i.e., people connected to only one task) for each task, which is rarely the case in real projects~\cite{lawrence1986organization}.

In fact, it can be verified that both measures assign a high bus-factor to a project represented by a graph of disconnected dyads (i.e., one person connected to one task).
We will experimentally demonstrate this flaw in Section~\ref{sec:experiments}.
Additionally, both measures rely on an arbitrary threshold $t$ (usually set to 0.5) to define the stalling condition of a project; consequently, the value of a project’s bus-factor depends on $t$.
Finally, MCS and MRS are not normalized, making them dependent on project size and hindering meaningful comparisons across projects.

Consequently, in the next section, building on our prior work~\cite{piccolo2024evaluating}, we introduce a novel definition of the bus-factor that overcomes the aforementioned weaknesses.

\subsection{A Better Measure: Bus-Factor as Network Robustness}
\label{sub:bf_net_robustness}

We base our measure on the maximum number of connected tasks, that is, the largest number of tasks contained in a single connected component.
As already illustrated in Fig.~\ref{fig:summary-of-the-paper}, using the maximum number of connected tasks allows us to capture the role of integrators, since it reflects both task coverage and connectedness.
Furthermore, by computing the area under the decay curve of the maximum number of connected tasks as people are progressively removed, we eliminate the need to define explicit stalling conditions, and thus remove the dependence on arbitrary thresholds (Section~\ref{sub:coverage_limitations}).
Finally, by normalizing this area by the theoretical maximum---obtained for a fully connected bipartite graph with the same number of people and tasks---we obtain a measure that can be interpreted in two ways: \emph{(1)} as the relative robustness of the project with respect to the maximum possible, and \emph{(2)} as the expected number of people whose removal would put the project at risk, obtained by multiplying the normalized bus-factor by the total number of people in the project.

Before defining our measure formally, we introduce the concept of a \emph{backbone set}.

\begin{definition}[Backbone set]
    Given a project represented as a bipartite graph $G = (P, T, E)$, a set $S \subseteq P$ is a backbone set if the induced bipartite subgraph $G[P \setminus S]$ consists of a collection of disconnected \emph{stars}, each centered on a single person (i.e., a vertex in $P$).
\end{definition}

Intuitively, a backbone set contains the integrators of a project.
The smaller the backbone set, the smaller the project’s bus-factor and, consequently, the more vulnerable the project.
However, we are not only interested in the cardinality of a project’s backbone set; rather, we focus on the maximum number of connected tasks as people in a backbone set are sequentially removed from the project.
Note that \emph{(1)} a set of disconnected stars is precisely the configuration for which MCS and MRS fail, and \emph{(2)} once $G$ has been reduced to a collection of disconnected stars, the maximum damage to the project is inflicted by removing people in decreasing order of their degree.

Let $\pi = (p_1, p_2, \dots, p_k)$ be a removal order such that the set $S = \{p \mid p \in \pi\}$ is a backbone set of size $k$.
Let $G^{\pi}_{i} = G[P \setminus \{p_1, p_2, \dots, p_i\}]$, with $G^{\pi}_{0} = G$, and let $\tau(G)$ denote the maximum number of connected tasks in $G$.
We measure the area under the decay curve of $\tau(G^{\pi}_{i})$ for $i = 1, \dots, k$ as
\begin{equation}
    \label{eq:robustness_pi}
    R(G, \pi) = \sum_{i=1}^{k} \tau \left(G^{\pi}_{i}\right).
\end{equation}

This definition allows us to generalize a widely used notion of network robustness to bipartite graphs (and hypergraphs) by defining robustness as the minimum possible area under the decay curve of the maximum number of connected tasks in $G$:
\begin{equation}
    \label{eq:robustness}
    \mathcal{R}(G) = \min_{\pi} R(G, \pi).
\end{equation}

We now show that computing bipartite network robustness, as defined in Eq.~\ref{eq:robustness}, is NP-hard.

\begin{theorem}
    \label{thm:network_robustness}
    \textsc{Bipartite Network Robustness} is NP-hard.
\end{theorem}

\begin{proof}
    We show that the decision version of \textsc{Bipartite Network Robustness} is NP-complete.
    An instance of \textsc{Bipartite Network Robustness} consists of a bipartite graph $G = (P, T, E)$ and two positive integers $k$ and $Q$, and asks whether there exists a removal sequence $\pi$ that is also a backbone set of length at most $k$ such that $R(G, \pi) \le Q$.
    
    First, note that \textsc{Bipartite Network Robustness} is in NP, since verifying a candidate solution $\pi$ requires computing $R(G, \pi)$, which can be done in polynomial time.
    In Section~\ref{sub:approx_net_robustness}, we will in fact provide a linear-time algorithm to compute $R(G, \pi)$.
    
    We show that \textsc{Bipartite Network Robustness} is NP-hard by reducing \textsc{Vertex Cover} to it.
    Given an instance of \textsc{Vertex Cover} on a graph $G' = (V', E')$ with parameter $k'$, we construct an instance of \textsc{Bipartite Network Robustness} as follows.
    We transform $G'$ into a bipartite graph $G = (P, T, E)$, where $P = V'$, $T = E'$, and $E = \{ (a,b) \mid a \in P,\ b \in T,\ \text{and } a \text{ is an endpoint of } b \}$ (the same transformation used for MCS, in Fig.~\ref{fig:Hardness_MCS_MRS}\textbf{A}).
    We set $k = k'$ and $Q = k|T|$.
    
    If $G'$ has a vertex cover $S$ of size at most $k'$, then removing the vertices in $S$ from $G'$ yields a graph consisting only of isolated vertices.
    Since our construction preserves incidence structure, removing the corresponding vertices from $P$ in $G$ results in a graph composed of disconnected stars, each centered on a vertex in $P \setminus S$ and connected to a number of vertices in $T$ equal to its degree in $G'$.
    Thus, $G$ admits a backbone set of size at most $k$.
    
    Conversely, suppose that $G$ admits a backbone set $\pi$ of size at most $k$ such that $R(G, \pi) \le Q$.
    The worst-case scenario for a backbone set of size $k$ occurs when people are removed in such a way that the maximum number of connected tasks never decreases throughout the removal process.
    In this case, $\tau(G^{\pi}_i) = |T|$ for all $i = 1, \dots, k$, and therefore $R(G, \pi) = k|T| = Q$.
    This implies that removing the vertices in $\pi$ must remove all edges in $G'$, and hence the corresponding vertices form a vertex cover of size at most $k'$ in $G'$.
    
    Therefore, $G$ admits a backbone set $\pi$ of size at most $k$ such that $R(G, \pi) \le Q$ if and only if $G'$ has a vertex cover of size at most $k'$, and thus \textsc{Vertex Cover} $\le_{\mathrm{p}}$ \textsc{Bipartite Network Robustness}.
\end{proof}

Our measure of robustness $\mathcal{R}(G)$ generalizes a widely used measure of network robustness for undirected graphs, introduced by Schneider \textit{et al.}~\cite{schneider2011mitigation}.
Zhao \textit{et al.}~\cite{zhao2021towards} noted that establishing the computational hardness of the robustness measure proposed by Schneider \textit{et al.} was still an open problem.
Theorem~\ref{thm:network_robustness} resolves this open problem.

The measure in Eq.~\ref{eq:robustness_pi} is conceptually related to the notion of \emph{natural connectivity} of a bipartite graph~\cite{shang2013measuring}, a spectral measure that captures the redundancy of alternative paths in a network.
Natural connectivity is a more appropriate robustness measure when flows and the length of paths must be taken into account.
In the case of the bus-factor, however, we do not deal with flows in the network; as such, measuring the maximum number of connected tasks is more appropriate.
In addition, computing natural connectivity is more expensive than computing the maximum number of connected tasks, since it depends on all the eigenvalues of the adjacency matrix of $G$.

Finally, we introduce our normalized measure of the bus-factor as the ratio between the robustness of the bipartite graph $G$, computed with a removal sequence $\pi \in \Pi(P)$ (where $\Pi(P)$ denotes the set of all possible permutations of the people in $P$) and the robustness of a fully connected bipartite graph of the same size:
\begin{equation}
    \mathcal{B}(G, \pi) =
    \frac{2}{m(2n - 1)}
    \sum_{i=1}^{n}
    \left[
        \tau\left(G^{\pi}_{i-1}\right) + \tau\left(G^{\pi}_{i}\right)
    \right],
    \quad \text{with } G^{\pi}_{0} = G.
\label{eq:normalised_bus_factor}
\end{equation}

Note that in Eq.~\ref{eq:normalised_bus_factor} we use a slightly different formulation from that of network robustness in Eq.~\ref{eq:robustness_pi}.
These are simply two different ways of estimating the area under the curve: Eq.~\ref{eq:robustness_pi} corresponds to a Gauss sum, whereas Eq.~\ref{eq:normalised_bus_factor} applies the trapezoidal rule, which provides a more accurate estimate of the area under a curve.
The normalized bus-factor $\mathcal{B}(G, \pi)$ takes values in $[0,1]$ and can be interpreted as the proportion of a project’s robustness relative to the theoretical maximum, thereby enabling comparisons across projects.
For instance, a bus-factor of $\mathcal{B}(G, \pi) = 0.5$ indicates that $G$ is half as robust as a fully connected graph of the same size.

Since computing the bus-factor of a project is NP-hard, in the next section we present efficient and scalable algorithms to approximate the bus-factor for all the definitions considered in our framework.

\section{The Practice of Approximating the Bus-Factor}
\label{sec:practice}

In this section, we present approximation algorithms for all the bus-factor formulations analyzed in the previous section, addressing both approximation quality and scalability.
For MCS and Robustness, we parameterize our algorithms by a removal order $\pi$, so that the approximation quality depends solely on $\pi$.
This design allows the same algorithms to be applied with alternative removal orders, should these lead to improved approximations.

\subsection{Approximating \textsc{Maximum Redundant Set}}
\label{sub:approx_redundant_set}

In Section~\ref{sub:framework_prior}, we showed that a redundant set is the complement of a partial set cover; therefore, any approximation algorithm for the minimum partial set cover yields an approximation algorithm for the Maximum Redundant Set.
The best possible polynomial-time approximation algorithm for Set Cover, unless P$=$NP, guarantees a solution that is at most a factor $O(\log m)$ larger than the optimal solution~\cite{cormen2022introduction}.
The same algorithm, with the same performance guarantees, can be used to approximate Partial Set Cover, since both problems can be seen as special cases of submodular set cover~\cite{wolsey1982analysis}.

Recall that MRS takes as input a graph $G = (P, T, E)$ and a real number $0 < t \le 1$, and returns the largest set of people that cover fewer than $t m$ tasks.
Conceptually, our approximation algorithm proceeds as follows: at each iteration, the person who covers the largest number of currently uncovered tasks is added to the set $S$, until more than $t m$ tasks are covered.
Since $S$ is an approximate solution to the minimum partial set cover, the set $X = P \setminus S$ constitutes an approximate solution to the MRS.
Adapting a technique from Lim \textit{et al.}~\cite{lim2014lazy}, we present a nearly linear-time approximation algorithm for MRS that scales to very large bipartite graphs (Algorithm~\ref{alg:max_redundant_set}).

\begin{algorithm}
\begin{algorithmic}[1]
\Function{Max-Redundant-Set}{$G,\; t$}
    \State $\mathcal{Q} \gets \mathrm{MaxPriorityQueue()}$ \Comment{Priority queue ordering people by task coverage}
    \State $\texttt{covered} \gets \mathrm{BitArray}(m)$ \Comment{Bit array indicating covered tasks; initially all bits are unset}
    \State $S = \emptyset$ \Comment{Stores the partial set cover}
    \ForAll{$p \in G\mathrm{.people()}$}
        \State $\mathcal{I}[p] \gets G\mathrm{.neighbors}(p)$ \Comment{Array tracking the tasks currently covered by each person}
        \State $\mathcal{Q}[p] \gets G\mathrm{.degree}(p)$
    \EndFor

    \While{$\texttt{covered}.\text{count}() < tm$}
        \State $p,\; \texttt{priority} \gets \mathcal{Q}\mathrm{.pop()}$
        \State $s \gets \mathcal{I}[p]$
        \State $\texttt{effective\_coverage} \gets s \setminus \texttt{covered}$
        \If{$|\texttt{effective\_coverage}| = \texttt{priority}$}
            \Comment{If the coverage of $p$ is up to date, add $p$ to $S$}
            \State $S \gets S \cup \{p\}$
            \State $\texttt{covered}\mathrm{.union}(\texttt{effective\_coverage})$
        \Else
            \Comment{Otherwise, update the coverage of $p$ and reinsert $p$ into $\mathcal{Q}$}
            \State $\mathcal{I}[p] \gets \texttt{effective\_coverage}$
            \State $\mathcal{Q}\mathrm{.push}(p,|\texttt{effective\_coverage}|)$
        \EndIf
    \EndWhile
    \State \textbf{return} $n - |S|,\; P \setminus S$ \Comment{MRS is the complement of a partial set cover}
\EndFunction
\end{algorithmic}
\caption{Approximation algorithm for the Maximum Redundant Set (MRS)}
\label{alg:max_redundant_set}
\end{algorithm}

In our algorithm, we use a priority queue $\mathcal{Q}$ to schedule the person with the highest task coverage, and we represent the set of $\texttt{covered}$ tasks and the set $S$ as bit arrays (Algorithm~\ref{alg:max_redundant_set}, lines~2--4).
For each person, we keep track of the uncovered tasks in a contiguous array of segments $\mathcal{I}$ and prioritize people by their degree (lines~5--8).
At each iteration, we pop a person $p$ from $\mathcal{Q}$ and check whether their priority is equal to the current number of uncovered tasks, $\texttt{effective\_coverage}$ (lines~9--13).
If so, we add $p$ to the partial set cover $S$ and update the set of covered tasks (lines~13--16); otherwise, we push $p$ back into the queue with updated priority and task coverage (lines~16--19).
When the number of covered tasks reaches at least $tm$, the loop terminates, and the algorithm returns the set $P \setminus S$, along with its cardinality, as the solution to MRS (lines~20--22).

Algorithm~\ref{alg:max_redundant_set} is an $O(n)$-approximation for the MRS.
The following theorem shows that this approximation ratio is the best achievable by a polynomial-time algorithm, unless $\mathrm{P}=\mathrm{NP}$.

\begin{theorem}
    \label{thm:aprox_max_redundant_set}
    Unless $\mathrm{P}=\mathrm{NP}$, it is NP-hard to approximate \textsc{Maximum Redundant Set} to within a factor $O(n^{1-\epsilon})$, for every constant $\epsilon > 0$, where $n$ is the number of people.
\end{theorem}

\begin{proof}[Proof (sketch)]
The Maximum Independent Set problem is the complement of the Minimum Vertex Cover problem~\cite{garey1979computers}.
Using the same transformation previously employed to prove the hardness of MCS and network robustness, we can transform an undirected, non-bipartite graph $G'$ into an equivalent bipartite graph $G$.
Computing the minimum set cover on $G$ is equivalent to computing the minimum vertex cover on $G'$~\cite{garey1979computers}; therefore, computing MRS on $G$ (with $t = 1$) is equivalent to computing the maximum independent set on $G'$.
However, unless $\mathrm{P}=\mathrm{NP}$, it is NP-hard to approximate the Maximum Independent Set problem to within a factor $O(n^{1-\epsilon})$ for every constant $\epsilon > 0$~\cite{zuckerman2006linear}.
Since MRS contains the Maximum Independent Set problem as a special case, this hardness result transfers to MRS.
\end{proof}

\subsection{Approximating \textsc{Minimum Critical Set}}
\label{sub:approx_critical_set}

To approximate the minimum critical set, we propose a node percolation approach in which people are progressively removed from the graph according to a removal order $\pi$.
We initialize an array containing the number of people covering each task, i.e., their degree (Algorithm~\ref{alg:min_critical_set}, lines~2--4).
As people are removed from $G$, we decrease the degree of the tasks to which they are connected (lines~6--9), update the count of isolated tasks when a task $\mathfrak{t}$ becomes isolated (lines~10--12), and track the number of people removed (lines~13--14).
When the number of isolated tasks exceeds $tm$ (with $t$ typically set to $0.5$~\cite{avelino2016novel}), the algorithm terminates, returning the cardinality of the MCS (lines~15--17).

\begin{algorithm}
\begin{algorithmic}[1]
\Function{Min-Critical-Set}{$G,\; t,\; \pi$}
    \ForAll{$\mathfrak{t} \in G\mathrm{.tasks()}$}
        \State $\texttt{task\_degs}[\mathfrak{t}] \gets G\mathrm{.degree}(\mathfrak{t})$ \Comment{Keeps track of task degrees as people are removed from $G$}
    \EndFor
    \State $\mathrm{mcs\_size, num\_isolated} \gets 0$ \Comment{Size of MCS, and number of isolated tasks}
    \While{$\mathrm{num\_isolated} \le tm$}
        \State $p \gets \pi \mathrm{.next()}$
        \ForAll{$\mathfrak{t} \in G\mathrm{.neighbors}(p)$}
            \State $\texttt{task\_degs}[\mathfrak{t}] \gets \texttt{task\_degs}[\mathfrak{t}] - 1$ \Comment{Decreases the coverage of the tasks covered by $p$}
            \If{\texttt{task\_degs}[$\mathfrak{t}$] = 0}
                \State $\mathrm{num\_isolated} \gets \mathrm{num\_isolated} + 1$ \Comment{Increases the count of isolated tasks}
            \EndIf
        \EndFor
        \State $\mathrm{mcs\_size} \gets \mathrm{mcs\_size} + 1$ \Comment{Increases the count of people removed so far (current size of the MCS)}
    \EndWhile
    \State \textbf{return} $\mathrm{mcs\_size}$
\EndFunction
\end{algorithmic}
\caption{Bus-factor as minimum critical set}\label{alg:min_critical_set}
\end{algorithm}

Algorithm~\ref{alg:min_critical_set} scales linearly in the number of edges since, at most, it has to process each person and their associated tasks, leading to a computational complexity of $O(|E|)$.
Since Algorithm~\ref{alg:min_critical_set} expects a removal order $\pi$ as input, the natural question is how to choose $\pi$.
A sensible choice is to remove people in decreasing order of degree (as done by Avelino \textit{et al.}~\cite{avelino2016novel}), which tends to work well in practice.
In fact, real-world networks are sparse and exhibit heavy-tailed degree distributions: relatively few nodes have high degree, while the vast majority have low degree~\cite{albert2002statistical}.
This characteristic implies that high-degree nodes are of paramount importance to the connectivity of the network.
For instance, Shang has shown that, in bipartite graphs, the removal of a relatively small fraction of high-degree nodes severely reduces the network’s \emph{natural connectivity}~\cite{shang2013measuring}.

Although removal by degree works well in practice for the aforementioned reasons, we now show that this removal order provides, in the worst case, an $O(n)$-approximation for the minimum critical set.

\begin{theorem}
    \label{thm:aprox_min_critical_set}
    Algorithm \ref{alg:min_critical_set}, when the removal order $\pi$ is defined in decreasing order of degree, is an $O(n)$-approximation for the \textsc{Minimum Critical Set} problem.
\end{theorem}

\begin{proof}
    We give a constructive proof of this theorem by constructing an instance for which Algorithm \ref{alg:min_critical_set} has an approximation ratio of $O(n)$.
    Let $ 0 < t < 1$, $n,m \in \mathbb{N}^+$, such that $n > tm + 1$.
    We create a graph $G = (P,T,E)$ as follows: the set $P$ contains $n$ nodes, the set $T$ contains $m$ nodes.
    We partition the set $T$ into two subsets $T_1$ and $T_2$ such that $|T_1| = tm + 1$ and $|T_2| = (1 - t)m - 1$, with $T_1 \cup T_2 = T$ and $T_1 \cap T_2 = \emptyset$.
    Similarly, we partition $P$ into two subsets $P_1$ and $P_2$ such that $|P_1| = tm + 1$ and $|P_2| = n - (tm + 1)$, with $P_1 \cup P_2 = P$ and $P_1 \cap P_2 = \emptyset$.
    We now define the edge set $E$.
    Each person in $P_1$ is connected to exactly one task in $T_1$, forming $tm + 1$ disconnected dyads.
    In addition, we construct a fully connected bipartite graph between the people in $P_2$ and the tasks in $T_2$.
    Let $O$ denote an optimal critical set and let $X$ denote the critical set returned by Algorithm~\ref{alg:min_critical_set} when $\pi$ removes people in decreasing order of degree.
    The optimal solution is $O = P_1$.
    However, the approximation algorithm removes people in $P_2$ first, since they have higher (or equal) degree than those in $P_1$.
    Consequently, $|X| \ge |P_2|$.
    As $t$ approaches zero, $|P_1|$ approaches one, while $|P_2|$ approaches $n - 1$.
    Therefore, $|X| = \Omega(n)\,|O|$, implying an approximation ratio of $O(n)$.
\end{proof}

Put differently, in the worst case, the removal order by people degree can return a minimum critical set as large as $n-1$ people, while the optimal value is as small as $1$.

\subsection{Approximating the Bus-Factor as \textsc{Bipartite Network Robustness}}
\label{sub:approx_net_robustness}
Computing Robustness requires, as people are removed from $G$, iterating over all connected components in $G$, counting the tasks they contain, and taking the maximum among them.
A naive implementation of this process, based on $n$ repeated applications of depth-first search (DFS), would lead to an algorithm that scales quadratically in the number of people, with computational complexity $O(n^2 m)$.

Here, we propose an efficient implementation of this process by tracking the maximum number of connected tasks using a Union–Find data structure: a disjoint-set structure in which each set is represented as a tree, and the root of the tree is the representative element of the set.
When two sets are merged, the roots of the corresponding trees are connected~\cite{cormen2022introduction}.
If, instead of removing people, we add people in the reverse order of removal, this data structure allows us to track the maximum number of connected tasks as components grow.
This idea leads to the following linear-time algorithm\footnote{The same algorithm can be used to compute the robustness of any arbitrary network.}, with a computational complexity of $O(|E|)$.

\begin{algorithm}
\begin{algorithmic}[1]
\Function{Bus-Factor}{$G,\; \pi$}
    \State $\mathcal{U} \gets \text{UnionFind}(G\mathrm{.num\_tasks()})$ 
        \Comment{A Union--Find structure for the $m$ tasks}
    \State $\texttt{tau} \gets [0]$ 
        \Comment{A vector storing the values $\tau\left(G_i^\pi\right)$}
    \ForAll{$p \in \pi\mathrm{.reverse()}$}
        \State $\texttt{neigs} \gets G\mathrm{.neighbors}(p) \cup \{p\}$ 
            \Comment{Gets the tasks to which $p$ is connected}
        \State $\mathcal{U}\mathrm{.union}(\texttt{neigs})$ 
            \Comment{Connects these tasks into a single component in $\mathcal{U}$}
        \State $\texttt{tau}\mathrm{.push}\!\left(
            \mathcal{U}\mathrm{.max\_num\_connected\_tasks}()
        \right)$ 
            \Comment{Computes the maximum number of connected tasks $\tau\left(G_i^\pi\right)$}
    \EndFor
    \State $\mathcal{B} \gets \frac{2}{m(2n - 1)} \sum_{i=1}^n \left(\texttt{tau}[i-1] + \texttt{tau}[i] \right)$ 
        \Comment{Computes Robustness according to Eq.~\ref{eq:normalised_bus_factor}}
    \State \textbf{return} $\mathcal{B}$
\EndFunction
\end{algorithmic}
\caption{Bus-factor as network robustness}
\label{alg:bus-factor}
\end{algorithm}

Algorithm~\ref{alg:bus-factor} takes as input the graph $G$ and a removal sequence $\pi$.
We initialize a Union–Find data structure $\mathcal{U}$ with $m$ disjoint tasks (line~2) and process the removal sequence $\pi$ in reverse order, from the last person to be removed to the first (line~4).
Each time we process a person $p$, we merge the components corresponding to the tasks to which $p$ is connected (lines~5--6) and store in the vector \texttt{tau} the maximum number of connected tasks in $\mathcal{U}$ (line~7).
Finally, the bus-factor $\mathcal{B}$ is computed by applying Eq.~\ref{eq:normalised_bus_factor} to the vector \texttt{tau}.

Removing people by degree again offers a good approximation of bipartite network robustness in practice.
Piccolo \textit{et al.} have empirically demonstrated the effectiveness of degree-based removal for studying the robustness of a real-world project involving the construction of a biomass power plant~\cite{piccolo2018design}.
Shang~\cite{shang2013measuring}, by contrast, verified through simulations the predictions of percolation theory on bipartite networks regarding the removal of high-degree nodes, and subsequently extended the theory by introducing the concept of \emph{anchor nodes}~\cite{shang2025percolation}, which are closely related to our notion of integrators.

Although removal by degree works well in practice, we shall prove that it yields an $O(n)$-approximation of network robustness in the worst case.

\begin{theorem}
    \label{thm:approx_bus_factor}
    The removal sequence $\pi$ defined in decreasing order of degree is an $O(n)$-approximation for the \textsc{Bipartite Network Robustness} problem.
\end{theorem}

\begin{proof}

    Similarly to what we did in Theorem~\ref{thm:aprox_min_critical_set}, we construct a graph for which removal by degree achieves an approximation ratio of $O(n)$.
    Let $k \in \mathbb{N}$, with $k \ge 2$.
    
    We build the graph as follows.
    We instantiate $k$ stars, each consisting of one person and $k+1$ tasks.
    Then, we add one additional person and connect this person to one task from each of the $k$ stars introduced previously.
    We now obtain a tree rooted at a central person, connected to $k$ tasks, with $k$ peripheral people, each connected to $k+1$ tasks.
    This tree consists of $n = k+1$ people and $m = k(k+1)$ tasks.
    
    The optimal removal sequence for this tree removes the central person first (degree $k$), followed by the peripheral people (degree $k+1$).
    Removing the central person reduces the network to $k$ stars, each with $k+1$ tasks.
    Each subsequent removal isolates $k+1$ tasks, until the last person is removed, at which point the graph consists of $m$ disconnected tasks.
    The total cost according to Eq.~\ref{eq:robustness_pi} for this sequence is $R(G, \pi^*)$ = $k(k+1) = k^2 + k$.

    Algorithm~\ref{alg:bus-factor}, with $\pi$ defined in decreasing order of degree, instead begins by removing the people with degree $k+1$, isolating $k$ tasks after each removal.
    After $k$ steps, the network is reduced to the central person connected to $k$ tasks.
    After removing the final person, the network becomes a collection of isolated tasks with no people, and therefore the number of tasks in the largest connected component is zero.
    The cost of this procedure is $R(G, \pi) = (m - k) + (m - 2k) + \dots + (m - k^2) = km - kS(k)$, where $S(k) = k(k+1)/2 = (k^2 + k)/2$. Therefore, $R(G, \pi) = k(k(k+1)) - k((k^2 + k)/2) = k^3 + k^2 - (k^3 + k^2)/2 = (k^3 + k^2)/2$.

    The approximation ratio is thus:
    \[
    \frac{R(G, \pi)}{R(G, \pi^*)} = \frac{k^3 + k^2}{2(k^2 + k)} = \frac{k^2(k+1)}{2k(k+1)} = \frac{k}{2} = \frac{n-1}{2} = O(n)
    \]
\end{proof}

We remark that Theorem~\ref{thm:approx_bus_factor} establishes an $O(n)$ worst-case approximation bound not only for the bus-factor, but for network robustness in general (including non-bipartite graphs).

While worst-case approximation bounds are informative, they say little about the \emph{typical} behavior of an approximation algorithm, or about its average performance on realistic graphs.
Therefore, in the next section, we empirically assess both the running time of our algorithms and the effectiveness of the degree-based removal strategy.

\section{Algorithmic Performance Evaluation}
\label{sec:performance}

We implement both our algorithms and bus-factor measures in Python and run all experiments on a laptop with a 2.9 GHz CPU and 32 GB of RAM.
We assess the running time of our algorithms on random bipartite networks with node sets ranging in size from $1\,000$ to $1\,000\,000$ nodes each, for a total number of edges between $400\,000$ and approximately $500\,000\,000$.
While random networks are not good models of real-world project networks~\cite{piccolo2018design}, they are well suited for measuring running time because they do not exhibit exploitable structure.
We set the threshold $t = 0.\overline{9}$ for both Maximum Redundant Set (hereafter MRS) and Minimum Critical Set (hereafter MCS) in order to force the algorithms to process the entire network.
We find that our algorithms scale linearly with the number of edges and remain efficient even on massive networks, processing half a billion edges in less than 20 seconds (Fig.~\ref{fig:performance}\textbf{A}).

Next, to assess the quality of the approximations beyond worst-case analysis, we compare our implementations against alternative processing and removal orders.
We use a random order as a baseline, since a reasonable approximation algorithm should outperform a random solution.
In addition to the random baseline, we compare against removal orders induced by betweenness centrality and eigenvector centrality, two strong baselines that are widely used in network robustness studies~\cite{artime2024robustness}.
We evaluate the relative size of the solutions for MRS and MCS: for MRS, larger solutions are better, while for MCS, smaller solutions are preferable.
For Robustness, instead, we evaluate the area under the curve $\mathcal{B}(G,\pi)$ obtained by defining $\pi$ according to each removal strategy: smaller areas correspond to better solutions.

For this comparison, we generate bipartite power-law networks, as they reproduce many properties of real-world project networks~\cite{piccolo2018design}.
Specifically, we sample the degree distributions of the two node sets from a power-law probability density function and generate graphs using the bipartite configuration model.
The power-law probability density function is defined as follows:
\begin{equation}
\label{eq:power_law}
f(x, \alpha, a, c) = \frac{\alpha}{c}\left( \frac{x - a}{c} \right)^{\alpha - 1},
\quad \text{with support } [a, a + c].
\end{equation}
\noindent
We sample networks with $7\,500$ people and $10\,000$ tasks.
People degrees are sampled from Eq.~\ref{eq:power_law} with parameters $\alpha = 0.2,\; a = 1,\; c = 100$, while task degrees are sampled with parameters $\alpha = 0.2,\; a = 1,\; c = 70$.

\begin{figure*}[t]
    \centering
    \includegraphics[width=0.99\linewidth]{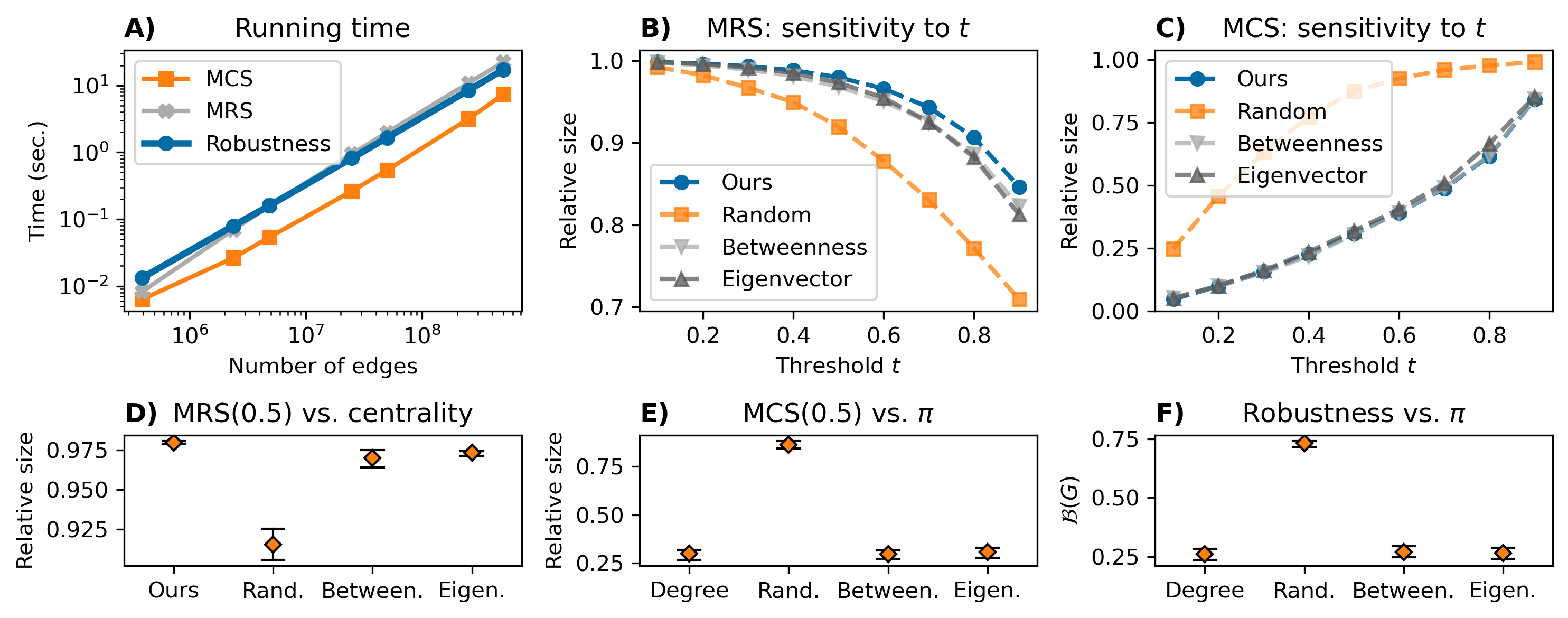}
    \caption{Algorithmic performance analysis.
\textbf{A)} Running time comparison.
\textbf{B)} Sensitivity of MRS to the threshold $t$.
\textbf{C)} Sensitivity of MCS to the threshold $t$.
\textbf{D)} Average relative size of the MRS for different processing strategies.
\textbf{E)} Average relative size of the MCS for different removal orders $(\pi)$.
\textbf{F)} Average robustness $\mathcal{B}(G,\pi)$ for different removal orders $(\pi)$.
\textbf{MRS}: Maximum Redundant Set;
\textbf{MCS}: Minimum Critical Set;
\textbf{Robustness}: Bus-factor as Bipartite Network Robustness.}
    \label{fig:performance}
\end{figure*}

This choice of parameters allows us to reproduce key characteristics of real-world project networks, namely: \emph{1)} tasks are typically more numerous than people; \emph{2)} degree distributions are right-skewed; and \emph{3)} only a small fraction of people (the integrators) is connected to a large number of tasks, while most people are specialists working on few tasks.
These characteristics have been observed across different types of projects, ranging from software development~\cite{yamashita2015pareto} to industrial design~\cite{piccolo2018design}.
We emphasize that there is nothing special about these specific parameters; we ran our experiments with different parameter values and obtained qualitatively similar results.

Since MRS and MCS depend on a threshold $t$, we first examined how their outputs vary as $t$ changes.
We observe that the size of MRS decreases monotonically with increasing $t$ (Fig.~\ref{fig:performance}\textbf{B}), while the size of MCS increases monotonically with $t$ (Fig.~\ref{fig:performance}\textbf{C}).
Across all values of $t$, our implementation consistently outperforms all baselines.
This monotonic behavior indicates that, qualitatively, the choice of $t$ does not substantially affect the trends observed.
However, from a quantitative standpoint, if MRS or MCS are to be interpreted as bus-factor measures, the choice of $t$ is crucial and requires careful justification.
In prior work, $t$ was typically set to $0.5$~\cite{ferreira2017comparison}.
Accordingly, for continuity with the existing literature, we set $t = 0.5$ in our subsequent analyses of MRS and MCS.

Finally, we evaluate the performance of our implementations against the baselines on $1\,000$ bipartite power-law networks.
For MRS (Fig.~\ref{fig:performance}\textbf{D}), our implementation produces solutions that are approximately $6\%$ larger than those obtained by the random baseline.
For MCS (Fig.~\ref{fig:performance}\textbf{E}) and Robustness (Fig.~\ref{fig:performance}\textbf{F}), our implementations yield solutions that are approximately three times smaller than those produced by the random baseline.
Moreover, our implementations perform at least as well as betweenness and eigenvector centrality–based strategies, while being computationally faster.

We now turn our attention to evaluating the suitability of MRS, MCS, and Robustness as bus-factor measures.

\section{Experimental Evaluation of Bus-Factor Measures}
\label{sec:experiments}

How can bus-factor measures be evaluated while maintaining a domain-agnostic perspective?
Empirical evaluations such as those conducted by Avelino \textit{et al.}~\cite{avelino2016novel} and Ferreira \textit{et al.}~\cite{ferreira2017comparison} are problematic in this respect: they are tailored to GitHub repositories and consider only a relatively small number of projects, most of which exhibit a trivial bus-factor of one.

Here, we take a different approach.
We propose a sensitivity analysis on synthetic power-law bipartite graphs (sampled from Eq.~\ref{eq:power_law}, as before), in which we apply controlled perturbations to evaluate the response of bus-factor measures in terms of relative change.
The perturbations we apply correspond to specific managerial actions that redistribute workload among contributors.

Concretely, we investigate the following research questions:

\begin{description}
    \item[Q1] How does the bus-factor change as the density of the network varies?

    To answer this question, we add or remove edges from the network, modeling situations in which people are assigned more or fewer tasks than in the original network.
    
    \item[Q2] How does the bus-factor change as people are added to the network?

    We add \emph{singletons} (i.e., people who work on only one task) and \emph{duplicates} (i.e., clones of existing contributors), modeling situations in which new people are added to a project in an attempt to increase redundancy and, consequently, robustness.

    \item[Q3] How does the bus-factor change as a consequence of reassigning people to tasks?

    Here, we exploit the known relationship between degree correlation and network robustness~\cite{newman2003mixing}.
    Specifically, we increase or decrease degree correlation in the network, without adding or removing nodes or edges, to induce corresponding changes in robustness.
\end{description}
\noindent
We address each of these questions in the following sections.

\subsection{Sensitivity of Bus-Factor Measures to Different Workload Distributions (Q1)}

Assigning more tasks to people is a common strategy to increase personnel knowledge within an organization, thereby making a project less vulnerable to contributor unavailability~\cite{piccolo2018design}.
Within our graph-theoretical framework, this corresponds to adding edges to the initial bipartite network, making the graph more connected and therefore more robust.
Conversely, removing edges models situations in which people are assigned, on average, fewer tasks, resulting in a less robust project structure.

We expect a well-designed bus-factor measure to be positively correlated with network density: the denser the network, the more robust the project; the sparser the network, the less robust it becomes.
We perform this set of simulations by adding or removing edges from the initial network in batches of $100$, for a total of $50\,000$ edges added or removed.

We find that MRS is largely insensitive to variations in network density, whereas MCS and Robustness qualitatively exhibit the expected behavior (Fig.~\ref{fig:densification_sparsification}).
However, this experiment already highlights the implications of the dependence of MCS on a fixed threshold.
As the network becomes denser, MCS saturates earlier than Robustness (Fig.~\ref{fig:densification_sparsification}\textbf{A}): during the first $25\,000$ edges added, MCS increases at a higher rate than Robustness and then plateaus (see the inset of Fig.~\ref{fig:densification_sparsification}\textbf{A}).
This behavior is a direct consequence of the threshold parameter $t$.

\begin{figure}[t]
    \centering
    \includegraphics[width=0.99\linewidth]{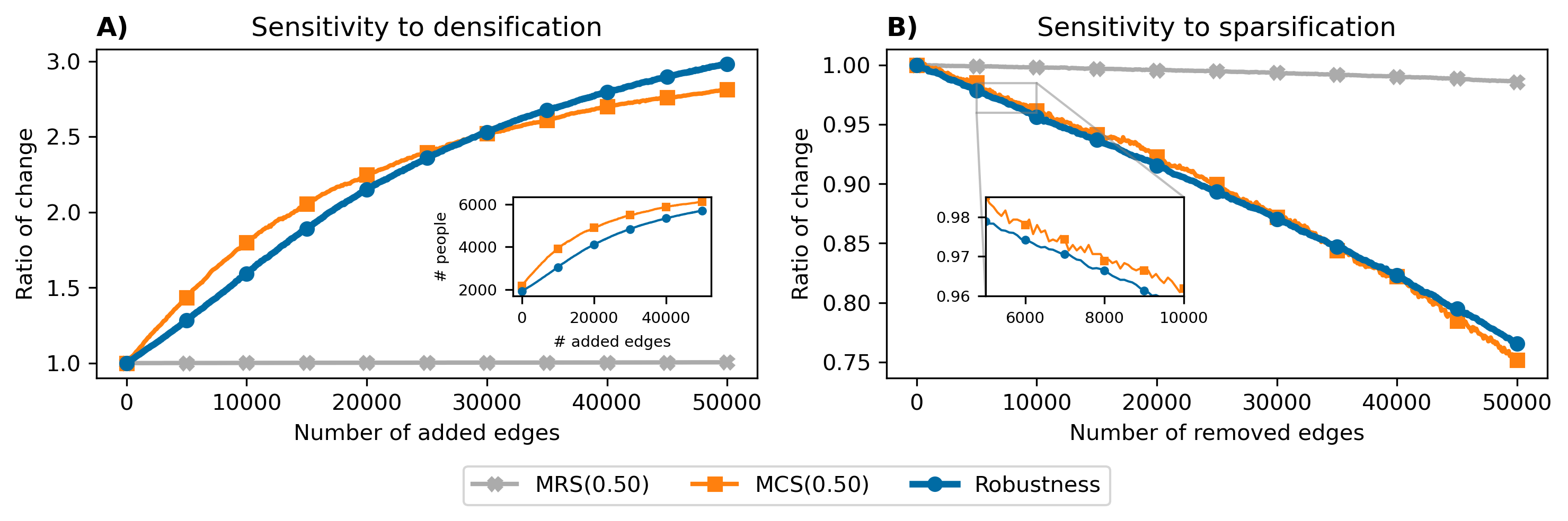}
    \caption{Sensitivity of bus-factor measures to changes in network density (\textbf{Q1}).
    \textbf{A)} Densification: MCS and Robustness increase with network density, as expected. Inset: values of MCS and Robustness expressed as number of people.
    \textbf{B)} Sparsification: MCS and Robustness decrease with network density, as expected. Inset: a zoomed-in view of the range $[5\,000, 10\,000]$ edges removed, showing that Robustness is more stable than MCS.
    \textbf{MRS}: Maximum Redundant Set;
    \textbf{MCS}: Minimum Critical Set;
    \textbf{Robustness}: Bus-Factor as Bipartite Network Robustness.}
    \label{fig:densification_sparsification}
\end{figure}

Similarly, while MCS and Robustness decrease at comparable rates as the network becomes sparser (Fig.~\ref{fig:densification_sparsification}\textbf{B}), MCS exhibits noticeable oscillations (see the inset of Fig.~\ref{fig:densification_sparsification}\textbf{B}).
These oscillations arise both from the use of a fixed threshold and from the fact that network sparsification alters node degrees, and consequently the removal order.
In contrast, Robustness exhibits a more stable behavior because it is based on connected components and the maximum number of connected tasks, rather than on task coverage alone.

Overall, this experiment already demonstrates the advantages of estimating a project’s bus-factor using Robustness over coverage-based alternatives.

\subsection{Sensitivity of Bus-Factor Measures to Strategies Aimed at Increasing Personnel Redundancy (Q2)}

We now turn to another strategy commonly used to strengthen projects and increase their robustness: hiring.
Adding people to a project can be an effective way to increase personnel redundancy and project robustness because, in addition to covering tasks, new people can improve the overall connectivity of the project by acting as integrators or anchors, thereby densifying sparser areas of the underlying network.

However, hiring is not a silver bullet, and if poorly implemented, this strategy can backfire.
First, the hiring process itself is complex, budget-constrained, and often slow, inefficient, or ineffective~\cite{behroozi2020debugging}.
Second, adding people to a project can impose higher coordination costs among personnel~\cite{brooks1995mythical}.
Finally, there is a substantial difference in terms of project robustness between adding integrators and adding specialists (i.e., people who work on only a few tasks): as already mentioned, integrators are the individuals who primarily increase the global robustness of a project~\cite{piccolo2018design}.

Here, we study the response of bus-factor measures to two strategies aimed at increasing redundancy:
\begin{enumerate}
    \item \emph{Adding singletons}: we add to the initial graph one singleton per task. A singleton is a specialist who works on only a single task.
    While adding singletons increases redundancy at the task level, it does not improve redundancy at the project level, as it fails to integrate distant modules of the network.
    \item \emph{Adding duplicates}: we add to the initial graph copies of existing people (duplicated nodes), starting from the person with the highest degree down to the person with the lowest degree.
    High-degree people are integrators, while low-degree people are specialists. This strategy increases redundancy up to a certain point -- namely, until integrators are duplicated -- after which it exhibits diminishing returns.
\end{enumerate}
\noindent
In line with empirical findings in project management~\cite{lawrence1986organization} and software engineering~\cite{majumder2019software}, a good bus-factor measure should be relatively insensitive to the addition of singletons and should correctly capture the diminishing returns associated with adding duplicates.

\begin{figure}[t]
    \centering
    \includegraphics[width=0.99\linewidth]{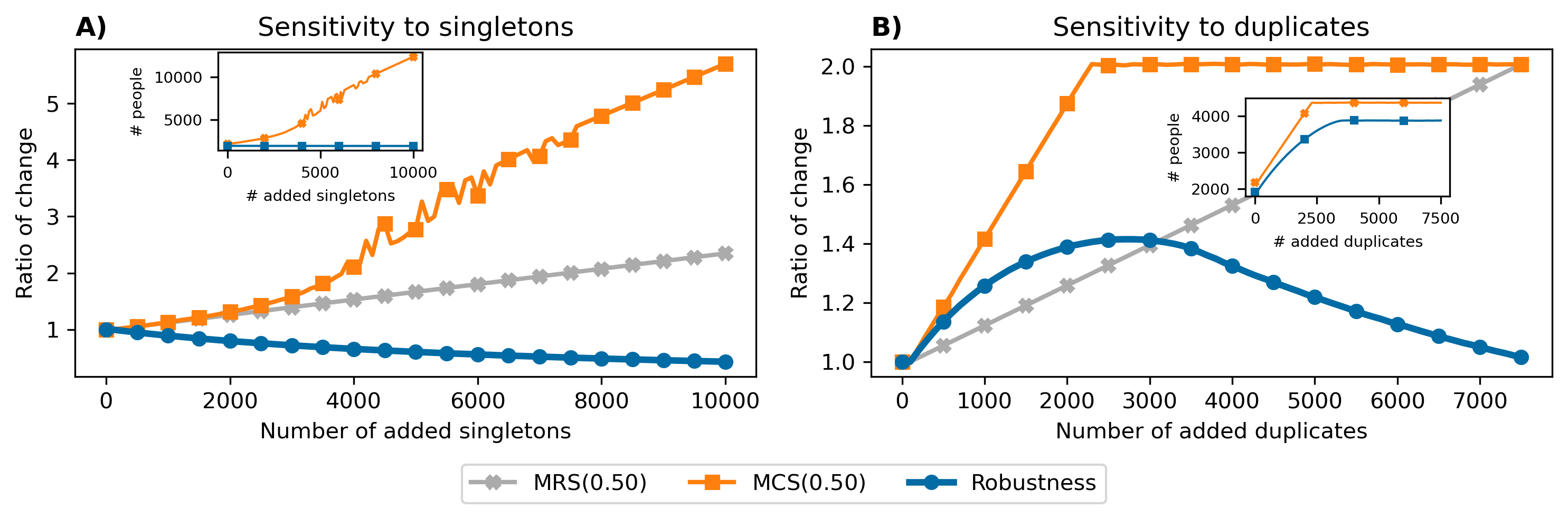}
    \caption{Sensitivity of bus-factor measures to strategies aimed at increasing personnel redundancy (\textbf{Q2}).
    \textbf{A)} Sensitivity to singletons (i.e., people who work on only one task). MCS and MRS increase as singletons are added to the network and do not exhibit an upper bound. This behavior is undesirable for a bus-factor measure. Robustness, instead, correctly captures the diminishing returns of adding singletons and decreases as singletons are added to the network. In the inset: values of MCS and Robustness expressed as numbers of people. While normalized Robustness decreases, reflecting diminishing returns, the number of people required to fragment the network remains constant, as expected.
    \textbf{B)} Sensitivity to duplicates. MRS grows indefinitely without an upper bound. MCS grows as duplicates are added until it saturates due to its threshold. Robustness, in line with expectations, captures the diminishing returns of adding people with progressively lower degree, exhibiting behavior consistent with empirical findings on the importance of integrators.
    \textbf{MRS}: Maximum Redundant Set;
    \textbf{MCS}: Minimum Critical Set;
    \textbf{Robustness}: Bus-Factor as Bipartite Network Robustness.}
    \label{fig:singletons_duplicates}
\end{figure}

We simulate these two strategies by adding singletons or duplicates in batches of $100$ people and report the results in Fig.~\ref{fig:singletons_duplicates}.
We find that both MRS and MCS are overly sensitive to the addition of singletons (Fig.~\ref{fig:singletons_duplicates}\textbf{A}).
Both measures grow without bound, implying that the bus-factor of a project could be increased indefinitely simply by adding singletons.
This result sharply contradicts empirical evidence on organizational differentiation \cite{lawrence1986organization}, coordination neglect in organizations \cite{heath2000coordination}, and iterative work dynamics \cite{piccolo2019iterations}, highlighting a major limitation of bus-factor measures based on MRS and MCS.

Robustness, instead, decreases as singletons are added to the project (Fig.~\ref{fig:singletons_duplicates}\textbf{A}).
This decrease is a direct consequence of both its reliance on the maximum number of connected tasks and its normalization factor (Eq.~\ref{eq:normalised_bus_factor}).
Since Robustness captures the degree of project fragmentation as people leave, singletons cannot increase its value.
Moreover, because Robustness is normalized, added singletons are effectively treated as wasted resources when compared to the contribution that an integrator could provide.
Consistently, the number of people required to disrupt the project remains stable throughout the simulation (inset of Fig.~\ref{fig:singletons_duplicates}\textbf{A}), further highlighting the flexibility and dual interpretability of Robustness.

MRS is also relatively insensitive to duplicates, while MCS saturates and fails to capture the diminishing returns associated with adding duplicates (Fig.~\ref{fig:singletons_duplicates}\textbf{B}).
Once again, Robustness successfully captures the diminishing returns provided by people with progressively lower degree, increasing more gradually than MCS and eventually decreasing when additional people no longer contribute meaningful redundancy -- a nuance that MCS is unable to represent.

Overall, this experiment demonstrates that MRS and MCS are ill-suited as bus-factor measures, as they overestimate the impact of singletons and underestimate the importance of integrators.
In contrast, Robustness accurately captures the effects of both singletons and integrators, exhibiting behavior consistent with theoretical expectations and empirical findings.

\subsection{Sensitivity of Bus-Factor Measures to Different Allocations of People to Tasks (Q3)}

Reassigning people to tasks, without hiring new personnel or changing the underlying workload distribution, is another effective strategy for increasing project robustness.
Piccolo \textit{et al.}~\cite{piccolo2024evaluating}, for instance, designed two heuristics that increased the bus-factor of a biomass power-plant project by $40\%$.

Here, we leverage a well-established result from network science that links network robustness to degree correlation (also known as degree assortativity)~\cite{newman2003mixing}.
Degree correlation measures the extent to which nodes with high (low) degree tend to connect to other nodes with high (low) degree.
It is defined as the Pearson correlation coefficient of the degrees at either end of an edge:
\begin{equation}
    r = \frac{\langle k_i k_j \rangle - \langle \tfrac{1}{2}(k_i + k_j) \rangle^2}
    {\langle \tfrac{1}{2}(k_i^2 + k_j^2) \rangle - \langle \tfrac{1}{2}(k_i + k_j) \rangle^2}
\end{equation}

\noindent
where $k_i$ and $k_j$ are the degrees of nodes $i$ and $j$ if they are connected by an edge, and $\langle \cdot \rangle$ denotes an average over all edges.
In general, higher degree correlation is associated with greater network robustness~\cite{newman2003mixing}, because high-degree nodes tend to connect with each other, forming dense and resilient sub-networks.

In this set of experiments, we assess the response of the bus-factor measures on synthetic networks with varying levels of degree correlation, while keeping the number of people, the number of tasks, and their degree distributions fixed.

Achieving this requires a non-trivial sampling procedure.
We sample the statistical ensemble $\{G\}$ according to the probability measure $\mu(G) \propto e^{-H(G)}$, induced by the Hamiltonian $H(G) = -J \sum_{(i,j) \in E} k_i k_j$, where $J$ is a parameter that controls the degree correlation of the sampled networks.

The sampling procedure consists of applying a Metropolis dynamics to the original graph $G$ by repeatedly rewiring two randomly selected edges $(u,z)$ and $(w,v)$, producing a new pair of edges (not already present in the network) $(u,v)$ and $(w,z)$.
For the resulting network $G'$, we compute the energy difference $\Delta H = H(G') - H(G)$ and accept $G'$ with probability $e^{-\Delta H}$.
We explore $17$ equally spaced values of $J$ in the interval $[-0.002, 0.002]$, generating $100$ independent networks for each value, for a total of $1\,700$ networks.

As shown in the inset of Fig.~\ref{fig:degree_correlation}\textbf{A}, the average degree correlation of the sampled networks increases monotonically with $J$.
From Fig.~\ref{fig:degree_correlation}\textbf{A}, we observe that MRS is again largely insensitive to changes in the allocation of people to tasks and even exhibits the opposite of the expected behavior, slightly decreasing as degree correlation increases.
Both MCS and Robustness, instead, behave consistently with theoretical expectations and prior findings from network science~\cite{newman2003mixing}.
\begin{figure}[t]
    \centering
    \includegraphics[width=0.99\linewidth]{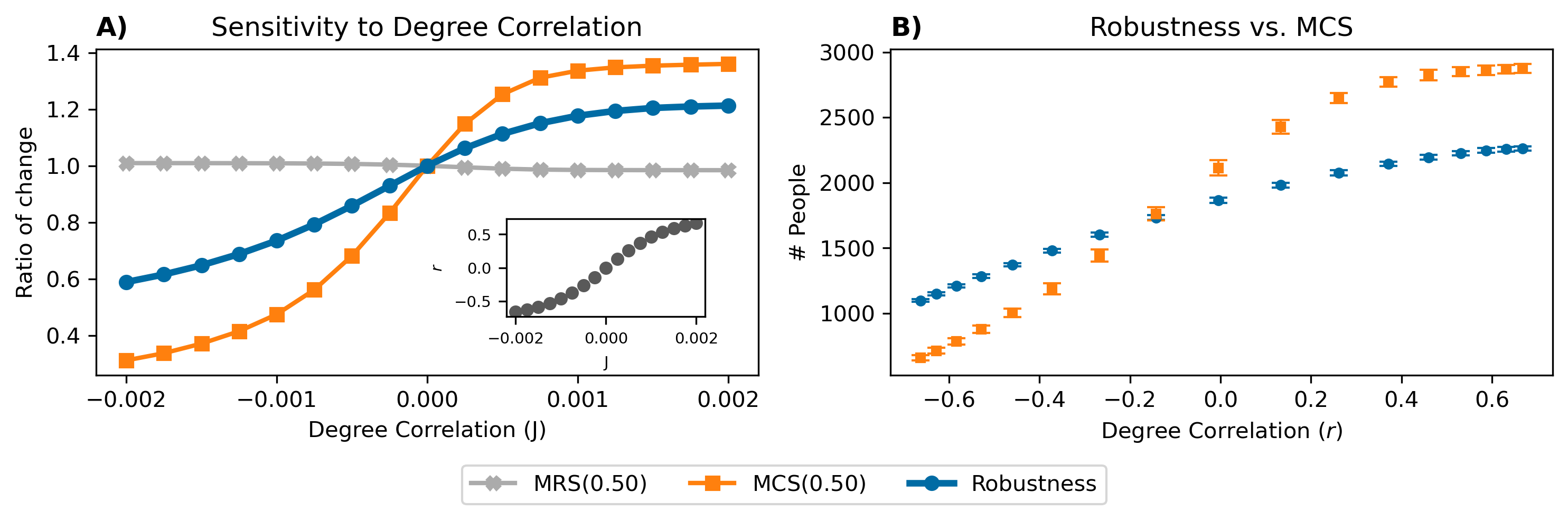}
    \caption{Sensitivity of bus-factor measures to changes in degree correlation (\textbf{Q3}).
    Each point represents an average over $100$ independent network realizations, for a total of $1\,700$ networks.
    \textbf{A)} MRS is insensitive to degree correlation, while MCS and Robustness exhibit the expected increasing behavior.
    \textbf{B)} Average bus-factor estimates expressed as numbers of people, with confidence intervals: Robustness displays a closer linear relationship with $r$ and lower variance than MCS.
    \textbf{MRS}: Maximum Redundant Set;
    \textbf{MCS}: Minimum Critical Set;
    \textbf{Robustness}: Bus-Factor as Bipartite Network Robustness.}
    \label{fig:degree_correlation}
\end{figure}

We note, however, that MCS introduces greater distortion with respect to degree correlation $r$ than Robustness (Fig.~\ref{fig:degree_correlation}\textbf{B}).
The curve produced by Robustness is closer to linear and exhibits lower overall variance than that produced by MCS.

\section{Discussion}
\label{sec:discussion}

At this point, we have a rather complete understanding of the theory that underpins the bus-factor and its approximation, as well as a detailed comparison of MRS, MCS, and Robustness as bus-factor measures.
What does our theory imply for network science and project management?
How do our results translate into actionable recommendations?
What are promising paths for future research?
In the following, we are going to answer these questions.

\subsection{Implications for Network Science and Project Management}
In this paper, we have developed a domain-agnostic graph theoretical framework that allows us to cast the problem of measuring the bus-factor as a combinatorial problem on bipartite graphs.
We defined three combinatorial problems that encode different definitions of bus-factor: maximum redundant set (MRS), minimum critical set (MCS), and bus-factor as network robustness (Robustness).
By showing that all three formulations are NP-hard, we have shown that computing the bus-factor of a project is not just a risk measure relevant to project management, but also a problem with important implications to the theory of computation: if we are able to solve any of these three problems in polynomial time, then $\text{P} = \text{NP}$.
Furthermore, our proof of NP-hardness of Robustness (Theorem~\ref{thm:network_robustness}) implies the NP-hardness of computing the network robustness measure $\mathcal{R}$ introduced by Schneider \textit{et al.}~\cite{schneider2011mitigation}, which was still an open problem~\cite{zhao2021towards}.

Our framework links network science, computer science, and project management. 
As this paper shows, project management problems can benefit directly from both the algorithmic and theoretical tools developed in network science and theoretical computer science.
For instance, percolation theory~\cite{albert2002statistical} provides theoretical justification for removing people in decreasing order of degree, while tools from approximation algorithms allow us to reason formally about the quality and limits of heuristic strategies. 
In this regard, our proofs of worst-case approximation bounds (Theorems~\ref{thm:aprox_max_redundant_set}, \ref{thm:aprox_min_critical_set}, \ref{thm:approx_bus_factor}) provide, to the best of our knowledge, the first approximation results for MRS, MCS, Robustness, and degree-based removal strategies.
To complete the picture, project management offers theories and empirical findings that we have used to inform our experimental comparison of bus-factor measures.

Our work provides extensive experimental and theoretical support for the importance of integrators as enablers of project robustness~\cite{piccolo2018design}, strengthening prior findings on software projects~\cite{majumder2019software}, and iterative development dynamics \cite{piccolo2019iterations}.
These results uncover close connections with both the concept of anchor nodes and the theory of robustness in bipartite graphs and hypergraphs~\cite{shang2025percolation}.
Moreover, we have shown a relation between project robustness and degree correlation~\cite{newman2003mixing}, which implies that it is possible to increase project robustness by reassigning people to tasks.
Finally, our extensive tests have shown that Robustness is the only measure that consistently exhibits the expected behavior, while coverage-based measures (MRS and MCS) have fatal flaws, notably their inability to capture project fragmentation, which render them unfit as bus-factor measures.

In short, the theory and practice of computing the bus-factor can be naturally framed as the theory and practice of computing bipartite network robustness (and vice versa).

\subsection{Translating our Findings to Practice}
\label{sub:use_bf}

Bus-factor measures have already been implemented in practice and shown to be informative for understanding project robustness and guiding corrective actions~\cite{jabrayildaze2022bus}.
Therefore, we would like to highlight here the key take-home messages of our analyses, which we believe provide actionable insights on estimating and improving the bus-factor.
\begin{enumerate}
    \item Robustness is currently the only bus-factor measure that is congruent with theoretical expectations and empirical findings. In addition, it is normalized, does not rely on a threshold, and it is able to capture project fragmentation, qualities that make it easy to use in practice.
    \item In the case where hiring is done with the purpose of increasing project robustness, the best strategy is to hire integrators that increase redundancy and connect project modules that are far apart.
    \item A project bus-factor can be improved by reallocating people to tasks with no need to increase the individual workload or to hire new people. This can be achieved through dedicated algorithms~\cite{piccolo2024evaluating}, or by increasing the degree correlation of the project's bipartite network, as we have shown.
    \item MCS is not appropriate as a bus-factor measure. However, it can be useful to estimate how unequal the workload distribution is. For instance, one can estimate the smaller set of people who has produced at least $80\%$ of a code base, testing whether the project follows the 80--20 Pareto distribution.
    \item MRS is uninformative as a bus-factor measure. However its complement, the partial set cover, can be useful to identify a small set of subject matter experts and integrators who can drive knowledge sharing practices and mentor new hires, increasing project sustainability. This links the bus-factor to the problem of finding the most influential nodes in a network.
\end{enumerate}
We remark that the use in practice of bus-factor measures as well as the above strategies should always be complemented by case-specific consideration such as people's skills and relative task importance.

\subsection{Avenues for Future Research}
\label{sub:future_research}

Our research paves the way to future extensions of the bus-factor, its theoretical underpinnings, and its approximation.
We sketch the way forward in the following:
\begin{enumerate}
    \item We have shown the NP-hardness of MRS, MCS, and Robustness and have proved that MRS is NP-hard to approximate within a factor $O(n^{1-\epsilon})$. The approximability of MCS and Robustness is still an open problem, the answer of which would give us important insights about the quality of approximation we can expect from \textit{any} heuristics.
    \item Our measure of bus-factor is only structural and does not account for relative task importance or people skills. Complementing the structural insights provided by Robustness with other social and technical factors~\cite{piccolo2019iterations} will likely offer a more complete understanding of project robustness.
    \item The theory of bus-factor can be extended by both developing its percolation theory and by uncovering what network measures correlate with it, beyond degree correlation.
    \item We have shown that the degree-based removal works well in practice, but algorithms that provide better estimations are very likely to exist. This is a promising research direction with potential to provide both better algorithms and higher understanding of the bus-factor.
\end{enumerate}
We hope that our work inspires more research at the intersection between network science, computer science, and project management.

\section{Conclusions}
\label{sec:conclusions}
The bus-factor is a measure of project robustness against personnel loss, turnover, and knowledge fragmentation. At its core, computing the bus-factor means answering a deceptively simple question: how many people must leave a project for it to stall?

The apparent simplicity of this question conceals substantial conceptual and computational difficulties. Defining when a project stalls, modeling how knowledge and responsibilities fragment as people leave, and formulating a measure that is independent of specific domains or platforms are all nontrivial challenges. As a result, existing approaches to bus-factor computation are largely ad hoc: they are project-specific or tailored to GitHub repositories, rely on simple coverage-based models, and depend on arbitrary thresholds to define failure.

In this paper, we introduced a domain-agnostic theoretical framework that casts bus-factor computation as a family of combinatorial problems on bipartite graphs. This framework allows existing approaches to be formalized and unified under two complementary formulations -- maximum redundant set (MRS) and minimum critical set (MCS) -- and makes their limitations explicit. Our analysis shows that, despite their intuitive appeal, coverage-based formulations fail to capture essential aspects of project robustness and are therefore inadequate as bus-factor measures.

Building on this framework, we proposed a novel bus-factor measure grounded in network robustness. Unlike prior definitions, this measure is normalized, threshold-less, and captures project degradation as a continuous process driven by fragmentation rather than a binary notion of failure. Both theoretical analysis and extensive empirical evaluation demonstrate that this robustness-based formulation aligns with expected behavior and consistently outperforms existing measures.

From a computational perspective, we showed that computing the bus-factor is NP-hard under all formulations considered in this work. This establishes fundamental limits on exact computation and motivates the use of approximation algorithms. Within this setting, we developed approximation strategies and proved worst-case guarantees for degree-based removal, providing the first formal approximation results for bus-factor computation.

Overall, this work elevates the bus-factor from an informal heuristic to a principled, rigorously grounded measure of project robustness. By reframing the bus-factor as a problem of bipartite network robustness, we provide a unifying theoretical foundation that connects project management, network science, and computational complexity, and enables systematic analysis, comparison, and approximation of robustness in collaborative systems.

\section*{Acknowledgments}
\noindent Sebastiano A. Piccolo is grateful to Marco Manna, Simona Perri and Aldo Ricioppo for helpful comments and fruitful conversations over NP-hardness and approximability of NP-hard optimization problems.
This research is partially supported by MUR under PNRR project PE0000013-FAIR, Spoke 9 - Green-aware AI -- WP9.2 and PN RIC project ASVIN ``Assistente Virtuale Intelligente di Negozio'' (CUP B29J24000200005).



\bibliographystyle{elsarticle-num} 
\bibliography{biblio}



\end{document}